\documentclass[a4paper,USenglish,cleveref,numberwithinsect,autoref]{lipics-v2021}

\newcommand{\M}{\mathcal{M}}
\newcommand{\N}{\mathcal{N}}
\newcommand{\Obs}{\mathcal{T}}
\newcommand{\Spec}{\mathcal{S}}

\newcommand{\U}{\mathcal{U}}

\newcommand{\converges}{\ensuremath{\mathord{\downarrow}}}
\newcommand{\diverges}{\ensuremath{\mathord{\uparrow}}}

\usepackage{soul}
\usepackage[T1]{fontenc}
\usepackage{lmodern}
\usepackage[utf8]{inputenc}
\usepackage{amsmath}
\usepackage{amssymb}
\usepackage{pbox}
\usepackage{tabularx}
\usepackage{color}
\usepackage{colortbl}
\usepackage{graphicx}
\usepackage{pifont}
\usepackage{booktabs}
\usepackage{microtype}
\usepackage{float}
\usepackage{stmaryrd}
\usepackage{bbding}
\usepackage[outline]{contour}
\contourlength{1.2pt}

\usepackage{tikz}
\usetikzlibrary{patterns, intersections, positioning, shapes.geometric, calc,
	automata, arrows, decorations.pathreplacing,decorations.pathmorphing}
\usetikzlibrary{cd}

\newcommand{\treeNodeLabel}[1]{\contour{white}{#1}}

\tikzset{
  initial text={},
	treenode/.style = {align=center, inner sep=0pt, text centered},
  basis/.style = {
    pattern=north east lines,
    pattern color=magenta!80!black!80!white,
  },
  frontier/.style={
    pattern=crosshatch dots,
    pattern color=yellow!80!black,
  },
  q0class/.style={
    pattern=vertical lines,
    pattern color=red!60!white,
  },
  q1class/.style={
    pattern=north west lines,
    pattern color=blue!60!white,
  },
  q2class/.style={
    pattern=crosshatch,
    pattern color=green!50!black!60!white,
  },
}

\usepackage{xspace}
\usepackage{multirow}
\usepackage{float}
\usepackage{algorithm,algpseudocode}
\usepackage{aliascnt}
\usepackage{etoolbox}

\usepackage[capitalise]{cleveref}

\newcommand{\partialto}{\ensuremath{\rightharpoonup}}

\newcommand{\apart}{\ensuremath{\mathrel{\#}}}

\newcommand{\semantics}[1]{\ensuremath{{\llbracket #1\rrbracket}}}
\newcommand{\stateCan}[2]{\ensuremath{\delta(#1,#2)\converges}}
\newcommand{\takeout}[1]{\relax}

\newenvironment{proofappendix}[2][Proof of]{%
	\paragraph*{#1~\autoref{#2}}%
	\addcontentsline{toc}{subsection}{#1~\autoref{#2}}
}{%
}

\newif\iflong
\longtrue% set to true for report version

\title{New Fault Domains for Conformance Testing of Finite State Machines}

\titlerunning{New Fault Domains for Conformance Testing}

\author{Frits Vaandrager}{Radboud University, Nijmegen, The Netherlands}{Frits.Vaandrager@ru.nl}{https://orcid.org/0000-0003-3955-1910}{} 
\author{Ivo Melse}{Radboud University, Nijmegen, The Netherlands}{Ivo.Melse@ru.nl}{}{}

\funding{Supported by NWO project OCENW.M.23.155, Evidence-Driven Black-Box Checking (EVI).}

\relatedversiondetails{This is the extended version of an article that will appear in the Proceedings of the 36th International Conference on Concurrency Theory (CONCUR 2025).
}

\acknowledgements{%
This article grew out of earlier work of the first author \cite{Vaandrager2024,VaandragerGRW22} and the bachelor thesis of the second author \cite{Melse2024}. An early version was also made available on arXiv \cite{abs-2410-19405}.	
Many thanks to Kirill Bogdanov, Paul  Fiter\u{a}u-Bro\c{s}tean, Dennis Hendriks, B\'{a}lint Kocsis, Bram Pellen, Jurriaan Rot, and the anonymous reviewers for their insightfull feedback on earlier versions.	
}

\authorrunning{F.\ Vaandrager and I.\ Melse}

\Copyright{Frits Vaandrager and Ivo Melse}

\ccsdesc[500]{Theory of computation}

\keywords{conformance testing, finite state machines, Mealy machines, apartness, observation tree, fault domains, $k$-$A$-complete test suites}

\nolinenumbers

\begin{document}
\maketitle

\begin{abstract}
	A fault domain reflects a tester's assumptions about faults that may occur in an implementation and that need to be detected during testing.
	A fault domain that has been widely studied in the literature on black-box  conformance testing is the class of finite state machines (FSMs) with at most $m$ states.  Numerous strategies for generating test suites have been proposed that guarantee fault coverage for this class. These so-called $m$-complete test suites grow exponentially in $m-n$, where $n$ is the number of states of the specification, so one can only run them for small values of $m-n$. 
	But the assumption that $m-n$ is small is not realistic in practice. In his seminal paper from 1964, Hennie raised the challenge to design checking experiments in which the number of states may increase appreciably.
	%Contribution 1
	In order to solve this long-standing open problem,
	we propose (much larger) fault domains that capture the assumption that all states in an implementation can be reached by first performing a sequence from some set $A$ (typically a state cover for the specification), followed by $k$ arbitrary inputs, for some small $k$.
	The number of states of FSMs in these fault domains grows exponentially in $k$.
	%Contribution 2
	We present a sufficient condition for \emph{$k$-$A$-completeness} of test suites with respect to these fault domains.
	%Contribution 3
	Our condition implies $k$-$A$-completeness of two prominent $m$-complete test suite generation strategies, the Wp and HSI methods. 
	Thus these strategies are complete for much larger fault domains than those for which they were originally designed, and thereby solve Hennie's challenge.
	%Contribution 4
	We show that three other prominent $m$-complete methods (H, SPY and SPYH) do not always generate $k$-$A$-complete test suites.
\end{abstract}

\newpage
\section{Introduction}
We revisit the classic problem of black-box conformance testing \cite{LeeY96} in 
a simple setting in which both specifications and implementations can be described as (deterministic, complete) finite state machines (FSMs), a.k.a.\ Mealy machines.
Ideally, given a specification FSM $\cal S$, a tester would like to have a finite set of tests $T$ that is complete in the sense that an implementation FSM $\M$ will pass all tests in $T$ if and only if $\M$ is equivalent to $S$.
Unfortunately, such a test suite does not exist: 
if $N$ is the number of inputs in the longest test in $T$ then an implementation $\M$ may behave like $\cal S$ for the first $N$ inputs, but differently from that point onwards.  Even though $\M$ is not equivalent to $\cal S$, it will pass all tests in $T$.
This motivates the use of a \emph{fault domain}, a collection of FSMs that reflects the tester's assumptions about faults that may occur in an implementation and that need to be detected during testing. The challenge then becomes to define a fault domain that captures \emph{realistic assumptions} about possible faults, but still allows for the design of \emph{sufficiently small} test suites that are complete for the fault domain and can be run within reasonable time.

A fault domain that has been widely studied is the set $\U_m$ of finite state machines (FSMs) with at most $m$ states.  Test suites that detect any fault in this class are called \emph{$m$-complete}.
The idea of $m$-complete test suites can be traced back to Moore~\cite{Mo56} and Hennie \cite{Hennie64}.
Numerous methods for constructing $m$-complete test suites have been proposed, for different types of transition system models, see for instance
\cite{Vas73,Ch78,YP90,FujiwaraEtAl91,PetrenkoYLD93,PHK95,LeeY96,DorofeevaEY05,PetrenkoY05,SPY12,SmeenkMVJ15,SouchaB18,BosJM19,ThesisJoshua,KrugerJR24,Vaandrager2024,GazdaH25}. We refer to \cite{LeeY96,DorofeevaEMCY10,ThesisJoshua} for overviews and further references.
The interest in $m$-complete test suites is somewhat surprising, given that in a black-box setting there is typically no sensible way to bound the number of possible states of an implementation to a small $m$.
After all, each additional Boolean variable in an implementation potentially doubles the number of extra states.
This is problematic in practice, since the size of $m$-complete test suites generated by existing methods grows exponentially in $m-n$, where $n$ is the number of states of the specification. Actually,
Moore~\cite{Mo56} was not aiming for practical methods and only described \emph{gedanken-experiments}.  
He introduced the example of combination lock machines, and was therefore well aware of the combinatorial explosions in $m$-complete test suites. 
In his seminal paper from 1964,
Hennie \cite{Hennie64} also observed the exponential blow-up in $m$-complete test suites and wrote ``Further work is needed before it will be practical to design checking experiments in which the number of states may increase appreciably.''
In two classic papers, Vasilevskii \cite{Vas73} and Chow \cite{Ch78} independently showed that $m$-complete test suites can be constructed with a size that is polynomial in the size of the specification FSM, for a fixed value of $k = m-n$.  Nevertheless, the test suites generated by their $W$-method grow exponentially in $k$, and so they did not solve the problem raised by Hennie \cite{Hennie64}. 
%Interestingly, Chow \cite{Ch78} argued that it is not unreasonable to assume a valid bound on the number of states in the implementation ``if we have spent adequate effort in analyzing the specification and constructing the design''. Thus, even though the W-method is commonly used for black-box testing, Chow actually assumed a white-box setting to justify his method.

In this article, we solve the long-standing open problem of Hennie \cite{Hennie64} as follows:
\begin{enumerate}
\item 
As an alternative for $m$-completeness,
we propose fault domains $\U_k^A$ that contain all FSMs in which any state can be reached by first performing a sequence from some set $A$ (typically a state cover for the specification), followed by $k$ arbitrary inputs, for some small $k$.
These fault domains contain FSMs with a number of extra states that grows exponentially in $k$.
%Analysis of a large collection of FSM models of implementations of the DTLS, TLS, SSH and EDHOC protocols suggests that the new fault domains $\U_k^A$ make sense from a practical perspective.
\item
Based on ideas from \cite{DorofeevaEY05,Vaandrager2024,VaandragerGRW22}, we present a sufficient condition for \emph{$k$-$A$-completeness} of test suites with respect to these fault domains, phrased entirely in terms of properties of their testing tree.
We present a $\Theta(N^2)$-time  algorithm to check this condition for a testing tree with $N$ states.
\item
We show that our sufficient condition implies $k$-$A$-completeness of two prominent approaches for test suite generation:
the Wp-method of Fujiwara et al \cite{FujiwaraEtAl91}, and
the HSI-method of Luo et al \cite{Luo1995} and Petrenko et al \cite{YP90,PetrenkoYLD93}.
The W-method of Vasilevskii \cite{Vas73} and Chow \cite{Ch78}, and the
UIOv-method of Chan et al \cite{ChanEtAl89} are instances of the Wp-method, 
and the ADS-method of Lee \& Yannakakis \cite{LYa94} and the hybrid ADS method of Smeenk et al \cite{SmeenkMVJ15} are instances of the HSI-method.
This means that
%, indirectly, 
$k$-$A$-completeness of these methods follows as well.
Hence these $m$-complete test suite generation methods are complete for much larger fault domains than those for which they were designed originally.
%These larger fault domains can be of practical interest, even for small values of $k=m-n$.
%However, not all test generation methods in the literature are $k$-$A$-complete. 
\item
We present counterexamples showing that three other prominent test generation methods, 
the H-method of Dorofeeva et al \cite{DorofeevaEY05},
the SPY-method of Sim\~{a}o, Petrenko and Yevtushenko \cite{SPY12} and the
SPYH-method of Soucha and Bogdanov  \cite{SouchaB18},
do not always generate $k$-$A$-complete test suites.
\end{enumerate}

The rest of this article is structured as follows. 
First, Section~\ref{sec:mealy} recalls some basic definitions 
regarding (partial) Mealy machines, observation trees, and test suites.
Section~\ref{test suite completeness} introduces $k$-$A$-complete test suites,
and shows how they strengthen the notion of $m$-completeness.
Next, we present our sufficient condition for $k$-$A$-completeness in Section~\ref{sec: main result}.  Based on this condition (and its proof), Section~\ref{sec:applications} establishes $k$-$A$-completeness of the Wp and HSI methods, and $m$-completeness of the H-method. Finally, Section~\ref{sec: discussion}, discusses implications of our results and directions for future research.
All proofs and some examples are deferred to appendices.

%Sections \ref{prel functions and sequences}-\ref{prel apartness} have overlap with \cite{VaandragerGRW22}, and Section~\ref{sec: main result} has overlap with \cite{Vaandrager2024}. 

\section{Preliminaries}
\label{sec:mealy}
In this section, we recall a number of key concepts that play a role in this article: partial functions, sequences, Mealy machines, observation trees, and test suites.  

\subsection{Partial Functions and Sequences}
\label{prel functions and sequences}

We write $f \colon X \partialto Y$ to denote that $f$ is a partial function from $X$
to $Y$ and write $f(x) \converges$ to mean that $f$ is defined on $x$, that is,
$\exists y \in Y \colon f(x)=y$, and conversely write $f(x)\diverges$ if $f$ is
undefined for $x$.
%, and $f(x) \uparrow$ if the result is undefined.
Often, we identify a partial function $f \colon X \rightharpoonup Y$ with the set $\{ (x,y) \in X \times Y \mid f(x)=y \}$. 
%The composition of partial
%maps $f\colon X\partialto Y$ and $g\colon Y\partialto Z$ is denoted by $g\circ
%f\colon X\partialto Z$, and we have $(g\circ f)(x)\converges$ iff $f(x)\converges$
%and $g(f(x)) \converges$.
%There is a partial order on
%$X\partialto Y$ defined by $f\sqsubseteq g$
%for $f,g\colon X\partialto Y$ iff for all $x\in X$, $f(x)\converges$ implies $g(x)\converges$ and
%$f(x) = g(x)$.
We use Kleene equality on partial functions, which states that on
a given argument either both functions are undefined, or both are defined and
their values on that argument are equal.

Throughout this paper, we fix a nonempty, finite set $I$ of \emph{inputs} and a set $O$ of \emph{outputs}. We use standard notations for sequences.  If $X$ is a set then $X^*$ denotes the set of finite \emph{sequences} (also called \emph{words}) over $X$. For $k$ a natural number, $X^{\leq k}$ denotes the set of sequences over $X$ with length at most $k$. We write $\epsilon$ to denote the empty sequence, $X^+$ for the set $X^* \setminus \{\epsilon\}$, $x$ to denote the sequence consisting of a single element $x \in X$, and $\sigma \cdot \rho$ (or simply $\sigma  \rho$) to denote the concatenation of two sequences $\sigma, \rho \in X^*$. The concatenation operation is extended to sets of sequences by pointwise extension.
We write $| \sigma |$ to denote the length of sequence $\sigma$.
For a sequence $\tau = \rho ~ \sigma$ we say that $\rho$ and $\sigma$ are a prefix and a suffix of $\tau$, respectively. We write $\rho \leq \tau$ iff $\rho$ is a prefix of $\tau$.
A set $W \subseteq X^*$ is \emph{prefix-closed} if any prefix of a word in $W$ is also in $W$, that is, for all $\rho, \tau \in X^*$ with $\rho\leq\tau$, $\tau \in W$ implies $\rho \in W$.
For $W \subseteq X^*$, $\mathit{Pref}(W)$ denotes the \emph{prefix-closure} of $W$, that is, the set $\{ \rho \in X^* \mid \exists \tau \in W : \rho \leq \tau \}$ of all prefixes of elements of $W$.
If $\sigma = x \rho$ is a word over $X$ with $x \in X$, then we write ${\sf hd}(\sigma)$ for $x$, and
${\sf tl}(\sigma)$ for $\rho$.

\subsection{Mealy machines}
\label{prel mealy}
Next, we recall the definition of Finite State Machines (FSMs) a.k.a.\  Mealy machines.

\begin{definition}[Mealy machine]
	\label{Mealy}
	A \emph{Mealy machine} is
	a tuple $\M = (Q, q_0, \delta, \lambda)$, where
%	\begin{itemize}[beginpenalty=99,topsep=1pt]
%	\item
	$Q$ is a finite set of \emph{states},
%	\item
	$q_0 \in Q$ is the \emph{initial state},
%	\item 
	$\delta\colon Q \times I \partialto Q$ is a (partial) \emph{transition function}, and
%	\item
$\lambda\colon Q \times I \partialto O$ is a (partial) \emph{output function} that satisfies
  $\lambda(q,i)\converges \Leftrightarrow \delta(q,i) \converges$, for $q \in Q$ and $i \in I$.
%	\end{itemize}
We use superscript $\M$ to disambiguate to which Mealy machine we refer, e.g.\  $Q^{\M}$, $q^{\M}_0$, $\delta^{\M}$ and $\lambda^{\M}$.
We write $q\xrightarrow{i/o}q'$ 
%for $q,q'\in Q$, $i\in I$, $o\in O$, 
to denote $\lambda(q,i) = o$ and $\delta(q,i) = q'$.
We call a state $q \in Q$ \emph{complete} iff $q$ has an outgoing transition for each input, that is, $\delta(q,i) \converges$, for all $i \in I$.
A set of states $W \subseteq Q$ is \emph{complete} iff each state in $W$ is complete.
The Mealy machine $\M$ is \emph{complete} iff $Q$ is complete.
The transition and output functions are lifted to sequences in the usual 
way. %, for $q \in Q$, $i \in I$ and $\sigma \in I^*$:
% \begin{eqnarray*}
% 	\delta(q, \epsilon) & = & q\\
% 	\delta(q, i \sigma) & = & \begin{cases}
% 	 \delta(\delta(q,i), \sigma) & \mbox{if } \delta(q,i) \downarrow\\
 %		\mbox{undefined} & \mbox{otherwise}
 %	\end{cases}\\
% 	\lambda(q, \epsilon) & = & \epsilon\\
% 	\lambda(q, i \sigma) & = & \begin{cases}
 %		\lambda(q,i) \cdot \lambda(\delta(q,i), \sigma) & \mbox{if } \delta(q,i) \downarrow\\
% 		\mbox{undefined} & \mbox{otherwise}
% 								\end{cases}
% \end{eqnarray*}
Let $q, q'\in Q$, $\sigma \in I^*$ and $\rho \in O^*$.
We write $q\xrightarrow{\sigma/\rho}q'$ to denote $\lambda(q,\sigma) = \rho$ and $\delta(q,\sigma) = q'$.
We write 
$q\xrightarrow{\sigma/\rho}$ if there is a $q' \in Q$ with $q\xrightarrow{\sigma/\rho}q'$,
we write
$q\xrightarrow{\sigma}q'$ if there is a $\rho\in O^*$ with $q\xrightarrow{\sigma/\rho}q'$, and
we write
$q\xrightarrow{+} q'$ if there is a $\sigma \in I^+$ with $q\xrightarrow{\sigma}q'$.
If $q_0 \xrightarrow{\sigma}q$ then we say that $q$ is \emph{reachable} via $\sigma$.

A \emph{state cover} for $\M$ is a finite, prefix-closed set of input sequences $A \subset I^*$ such that, for every $q \in Q$, there is a $\sigma \in A$ such that $q$ is reachable via $\sigma$.
A state cover $A$ is \emph{minimal} if each state of $\M$ is reached by exactly one sequence from $A$.
%We say that 
$\M$ is \emph{initially connected} if it has a state cover.
We will only consider Mealy machines that are initially connected.
\end{definition}

\begin{definition}[Semantics and minimality]
	\label{def:equivalence}
The \emph{semantics} of a state $q$ of a Mealy machine $\M$ is the map $\semantics{q}^{\M} \colon I^*\partialto O^*$
  defined by $\semantics{q}^{\M}(\sigma) = \lambda^{\M}(q, \sigma)$.
		
States $q, r$ of Mealy machines $\M$ and $\N$, respectively, are \emph{equivalent},
written $q \approx r$, iff $\semantics{q}^{\M} = \semantics{r}^{\N}$.
Mealy machines $\M$ and $\N$ are \emph{equivalent}, written $\M \approx \N$, iff their initial
states are equivalent: $q_0^\M \approx q_0^\N$.
A Mealy machine $\M$ is \emph{minimal} iff, for all pairs of states $q, q'$, $q \approx q'$ iff $q =q'$.
\end{definition}

\begin{example}
Figure~\ref{Fig:Conformance} shows an example (taken from \cite{ThesisJoshua}) with a graphical representation of two minimal, complete Mealy machines that are inequivalent, since the input sequence $a b a$ triggers different output sequences in both machines.
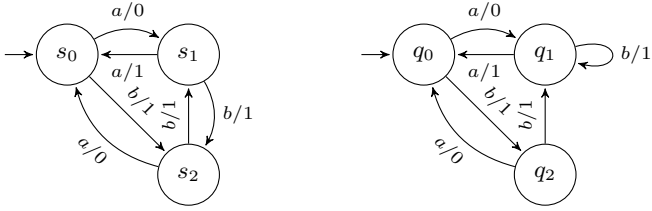
\begin{figure}[t!]
	\begin{center}
		%\vspace{-1cm}
		\begin{tikzpicture}[->,>=stealth',shorten >=1pt,auto,node distance=1.8cm,main node/.style={circle,draw,font=\sffamily\large\bfseries},
		]
		\def\yoffset{8mm}
		\node[initial,state,frontier] (0) {\treeNodeLabel{$s_0$}};
		\node[state,frontier] (1) [right of=0] {\treeNodeLabel{$s_1$}};
		\node[state,frontier] (2) [below of=1] {\treeNodeLabel{$s_2$}};
		
		\node[initial,state,frontier] [right of=1,xshift=1.5cm](q0) {\treeNodeLabel{$q_0$}};
		\node[state,frontier] (q1) [right of=q0] {\treeNodeLabel{$q_1$}};
		\node[state,frontier] (q2) [below of=q1] {\treeNodeLabel{$q_2$}};
		
		\path[every node/.style={font=\sffamily\scriptsize}]
		(0) edge [bend left] node[sloped,above] {$a/0$} (1)
		edge node[sloped,above] {$b/1$} (2)
		(1) edge node[below] {$a/1$} (0)
		edge [bend left] node[right] {$b/1$} (2)
		(2) edge [bend left] node[sloped,below] {$a/0$} (0)
		edge  node[sloped] {$b/1$} (1)
		(q0) edge node[sloped, above] {$b/1$} (q2)
		edge [bend left] node {$a/0$} (q1)	
		(q1) edge node[below] {$a/1$} (q0)
		edge [loop right] node {$b/1$} (q1)
		(q2) edge [bend left] node[sloped,below] {$a/0$} (q0)
		edge node[sloped] {$b/1$} (q1);
		\end{tikzpicture}
		\caption{A specification $\cal S$ (left) and an inequivalent implementation $\M$ (right).}
		\label{Fig:Conformance}
	\end{center}
\end{figure}	
\end{example}

\subsection{Observation Trees}
\label{prel apartness}

%In our testing setting, an \emph{undefined} value in the partial transition map
%represents lack of knowledge. 
A \emph{functional simulation} is a function between Mealy machines that preserves the initial state and the transition and output functions.

\begin{definition}[Simulation]
\label{def refinement}
A \emph{functional simulation} between Mealy machines $\M$ and $\N$ is a function
$f\colon Q^{\M} \to Q^{\N}$ satisfying $f(q^{\M}_0)= q^{\N}_0$ 
and, for $q \in Q^{\M}$ and $i \in I$, 
\begin{eqnarray*}
	\delta^{\M}(q,i) \converges & ~~ \Rightarrow ~~ & f(\delta^{\M}(q,i)) = \delta^{\N}(f(q),i) \mbox{ and }
	\lambda^{\M}(q,i) = \lambda^{\N}(f(q), i).
\end{eqnarray*} 
We write $f\colon \M\to \N$ if $f$ is a functional simulation between $\M$ and $\N$.
\end{definition}
Note that if $f\colon \M\to \N$, each transition $q \xrightarrow{i/o} q'$ of $\M$ can be matched by a transition $f(q) \xrightarrow{i/o} f(q')$ of $\N$. 

%\begin{proof}
%	See \ref{sc:proofs}.
%\end{proof}

%\begin{lemma}
%	\label{la:refinement}
%  For a functional simulation $f\colon \M\to \N$ and
%  $q\in Q^\M$, we have $\semantics{q}\sqsubseteq \semantics{f(q)}$.
  % $\delta^{\M}(q^{\M}_0,\sigma) \defined$ implies
  % $\delta^{\N}(q^{\N}_0,\sigma) \defined$ and $\lambda^{\M}(q^{\M}_0, \sigma)
  % = \lambda^{\N}(q^{\N}_0, \sigma)$, for $\sigma \in I^*$.
%\end{lemma}

For a given Mealy machine $\M$, an \emph{observation tree} for $\M$ is a Mealy machine itself that
represents the inputs and outputs that have been observed during testing of $\M$. 
Using functional simulations, we may define it formally as follows.

\begin{definition}[Observation tree]
	A Mealy machine $\Obs$ is 
	a \emph{tree} iff for each $q \in Q^{\Obs}$ there is a unique $\sigma \in I^*$ s.t.\ $q$ is reachable via $\sigma$.
	We write $\mathsf{access}(q)$ for the sequence of inputs leading to $q$.
	For $U \subseteq Q^{\Obs}$, we define $\mathsf{access}(U) = \{ \mathsf{access}(q) \mid q \in U \}$.
	For $q \neq q^{\Obs}_0$, we write $\mathsf{parent}(q)$ for the unique state $q'$ with an outgoing transition to $q$.	
 A tree $\Obs$ is an \emph{observation tree} for a Mealy machine $\M$ iff
  there is a functional simulation $f$ from $\Obs$ to $\M$.
\end{definition}

\begin{example}
Figure~\ref{Fig:turnstile}(right) shows an observation tree $\Obs$ for the Mealy machine $\Spec$ of Figure~\ref{Fig:turnstile}(left).
Mealy machine $\Spec$, an example taken from \cite{SouchaB18}, models the behavior of a turnstile. Initially, the turnstile is locked ($L$), but when a coin is inserted ($c$) then, although no response is observed ($N$), the machine becomes unlocked ($U$). When a user pushes the bar ($p$) in the initial state, the turnstile is locked ($L$), but when the bar is pushed in the unlocked state it is free ($F$) and the user may pass. State colors indicate the functional simulation from $\Obs$ to $\Spec$.
\begin{figure}[t!]
	\begin{center}
		\begin{tikzpicture}[->,>=stealth',shorten >=1pt,auto,node distance=1.6cm,main node/.style={circle,draw,font=\sffamily\large\bfseries},
		]
		\def\yoffset{8mm}
		\node[initial,state,q1class] (L) {\treeNodeLabel{$L$}};
		\node[state,q2class] (U) [right of=0] {\treeNodeLabel{$U$}};
		\node[initial,state,q1class] (0) [right of=L, xshift=5cm]  {\treeNodeLabel{$0$}};
		\node[state,q2class] (1) [below left of=0,xshift=-1cm] {\treeNodeLabel{$1$}};
		\node[state,q2class] (2) [below left of=1] {\treeNodeLabel{$2$}};
		\node[state,q2class] (3) [below left of=2] {\treeNodeLabel{$3$}};
		\node[state,q1class] (4) [below of=3] {\treeNodeLabel{$4$}};
		\node[state,q1class] (5) [below right of=2,] {\treeNodeLabel{$5$}};
		\node[state,q1class] (6) [below of=5] {\treeNodeLabel{$6$}};
		\node[state,q1class] (7) [below right of=1] {\treeNodeLabel{$7$}};
		\node[state,q1class] (8) [below of=7] {\treeNodeLabel{$8$}};
		\node[state,q1class] (9) [below of=8] {\treeNodeLabel{$9$}};
		\node[state,q1class] (10) [below right of=0,xshift=1cm] {\treeNodeLabel{$10$}};
		\node[state,q2class] (11) [below left of=10] {\treeNodeLabel{$11$}};
		\node[state,q1class] (12) [below of=11] {\treeNodeLabel{$12$}};
		\node[state,q2class] (13) [below of=12] {\treeNodeLabel{$13$}};
		\node[state,q1class] (14) [right of=13] {\treeNodeLabel{$14$}};		
		\node[state,q1class] (15) [below right of=10] {\treeNodeLabel{$15$}};
		\node[state,q1class] (16) [below of=15] {\treeNodeLabel{$16$}};
		
		\path[every node/.style={font=\sffamily\scriptsize}]
		(L) edge [bend left] node[sloped,above] {$c/N$} (U)
		edge [loop below] node[left]{$p/L$} (L)
		(U) edge [bend left] node[below] {$p/F$} (L)
		edge [loop below] node[right] {$c/N$} (U)
		(0) edge node {$c/N$} (1)
		edge node {$p/L$} (10)
		(1) edge node {$c/N$} (2)
		edge node {$p/F$} (7)
		(2) edge node {$c/N$} (3)
		edge node {$p/F$} (5)
		(3) edge node[left] {$p/F$} (4)
		(5) edge node[left] {$p/L$} (6)
		(7) edge node {$p/L$} (8)
		(8) edge node {$p/L$} (9)
		(10) edge node {$c/N$} (11)
		edge node {$p/L$} (15)
		(11) edge node {$p/F$} (12)
		(12) edge node {$c/N$} (13)
		(13) edge node {$p/F$} (14)
		(15) edge node {$p/L$} (16)
		;
		\end{tikzpicture}
		\caption{A Mealy machine $\Spec$ (left) and an observation tree $\Obs$ for $\Spec$ (right).}
		\label{Fig:turnstile}
	\end{center}
\end{figure}
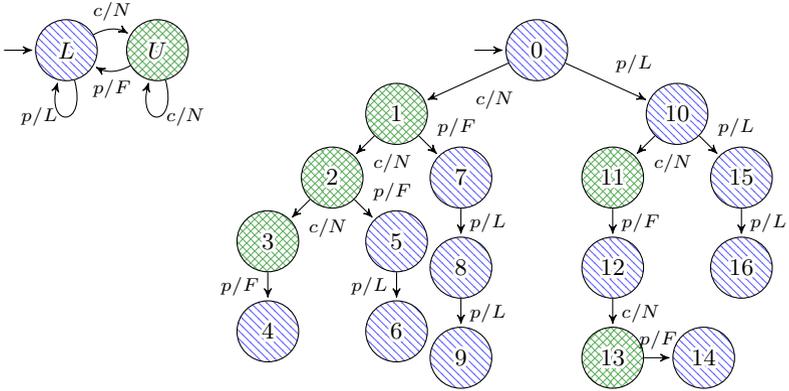
\end{example}

\subsection{Test Suites}
We recall some basic vocabulary of conformance testing for Mealy machines.

\begin{definition}[Test suites]
	Let $\Spec$ be a Mealy machine.
	A sequence $\sigma \in I^*$ with $\delta^{\Spec}(q_0, \sigma) \converges$ is called a \emph{test case} (or simply a \emph{test}) for $\Spec$.
	A \emph{test suite} $T$ for $\Spec$ is a finite set of tests for $\Spec$.
	A Mealy machine $\M$ \emph{passes} test $\sigma$ for $\Spec$ iff
	$\lambda^{\M} (q_0^{\M}, \sigma) = \lambda^{\Spec} (q_0^{\Spec}, \sigma)$, and
	passes  test suite $T$ for $\Spec$ iff it passes all tests in $T$.
\end{definition}
Observe that when $\M$ passes a test $\sigma$, it also passes all prefixes of $\sigma$.  This means that only the maximal tests from a test suite $T$ (tests that are not a proper prefix of another test in $T$) need to be executed to determine whether $\M$ passes $T$.
Also note that when $\M$ passes a test $\sigma$, we may conclude $\delta^{\M}(q_0^{\M}, \sigma) \converges$.

We like to think of test suites as observation trees. 
Thus, for instance, the test suite $T = \{ cccp, ccpp, cppp, pcpcp, ppp \}$ for the Mealy machine $\Spec$  of Figure~\ref{Fig:turnstile}(left) corresponds to the observation tree of Figure~\ref{Fig:turnstile}(right).
The definition below describes the general procedure for constructing a testing tree for a given test suite $T$ for a specification $\Spec$.  The states of the testing tree are simply all the prefixes of tests in $T$.  Since $T$ may be empty but a tree needs to have at least one state, we require that the empty sequence $\epsilon$ is a state.

\begin{definition}[Testing tree]
	\label{def: testing tree}
	Suppose $T$ is a test suite for a Mealy machine $\Spec$.
	Then the \emph{testing tree} $\mathsf{Tree}(\Spec, T)$ is the observation tree $\Obs$ given by:
	\begin{itemize}
		\item 
		$Q^{\Obs} = \{\epsilon\} \cup \mathit{Pref}(T)$
		and
		$q_0^{\Obs} = \epsilon$,
		\item 
		For all $\sigma \in I^*$ and $i \in I$ with $\sigma i \in Q^{\Obs}$, $\delta^{\Obs}(\sigma, i) = \sigma i$,
		\item 
		For all $\sigma \in I^*$ and $i \in I$ with $\sigma i \in Q^{\Obs}$, $\lambda^{\Obs}(\sigma, i) = \lambda^{\Spec}( \delta^{\Spec} (q_0^{\Spec}, \sigma), i)$.
	\end{itemize}
\end{definition}

There is a functional simulation from a testing tree to the specification that was used during its construction.

\begin{lemma}
	\label{lemma testing tree}
	 The function $f$ that maps each state $\sigma$ of $\Obs = \mathsf{Tree}(\Spec, T)$ to the state $\delta^{\Spec} (q_0^{\Spec}, \sigma)$ of $\Spec$ is a functional simulation.
\end{lemma}
%\begin{proof}
%	See \ref{sc:proofs}.
%\end{proof}

The next lemma, which follows from the definitions, illustrates the usefulness of testing trees: a Mealy machine $\M$ passes a test suite $T$ for $\Spec$ iff there exists a functional simulation from $\mathsf{Tree}(\Spec, T)$ to $\M$.

\begin{lemma}
	\label{observation tree when test suite passes}
		Suppose $\Spec$ and $\M$ are Mealy machines, $T$ is a test suite for $\Spec$, and
		$\Obs = \mathsf{Tree}(\Spec, T)$. Suppose function $f$ maps each state $\sigma$ of $\Obs$ to state $\delta^{\M} (q_0^{\M}, \sigma)$ of $\M$.
		Then $f \colon \Obs \to \M$ iff $\M$ passes $T$.
\end{lemma}
%\begin{proof}
%	See \ref{sc:proofs}.
%\end{proof}

%Note that test suites with the same prefix-closure correspond to the same testing tree.
%There is a one-to-one correspondence between observation trees (up to isomorphism) and finite test suites $T$ that do not contain redundant prefixes (i.e., if $\sigma, \rho \in T$ with $\sigma$ a prefix of $\rho$ then we require $\sigma=\rho$.)
%\begin{assumption}
%We implicitly require that via tests, the observation tree $\Obs$ is gradually extended.
%\end{assumption}

%Suppose $\Obs$ is an observation tree for implementation $\M$.  If all tests, as recorded in the  tree, have passed then $\Obs$ is also an observation tree for $\Spec$.
%However, as soon as a test fails then $\Obs$ is no longer an observation tree for $\Spec$: there exists no functional simulation from $\Obs$ to $\Spec$, since the observation tree contains a state $q$ such that the output in response to input sequence $\mathsf{access}(q)$ is different for $\Spec$ and $\M$.

\section{Fault Domains and Test Suite Completeness}
\label{test suite completeness}

%In this section, we propose alternative notions of completeness for test suites. 
%For any (finite) test suite $T$ for $\Spec$, we can trivially construct a faulty implementation that passes $T$.
%This justifies the following definition of 
A \emph{fault domain} reflects the tester's assumptions about faults that may occur in an implementation and that need to be detected during testing.

\begin{definition}[Fault domains and $\U$-completeness]
	A \emph{fault domain} is a set $\U$ of Mealy machines.
	A test suite $T$ for a Mealy machine $\Spec$ is \emph{$\U$-complete} if, for each $\M \in \U$,
	$\M$ passes $T$ implies $\M \approx \Spec$.
\end{definition}

The next three lemmas are immediate consequences of the above definition.
In particular, whenever $T$ is $\U$-complete, completeness is preserved when we add tests to $T$ or remove machines from $\U$.

\begin{lemma}
\label{extending test suite preserves completeness}
If $T$ and $T'$ are test suites with $T \subseteq T'$ and $T$ is $\U$-complete,
then $T'$ is $\U$-complete.
\end{lemma}

\begin{lemma}
	\label{completeness inclusion}
	If $\U$ and $\U'$ are fault domains with $\U' \subseteq \U$, and $T$ is a $\U$-complete test suite, then $T$ is $\U'$-complete.
\end{lemma}

\begin{lemma}
	\label{completeness union}
	If a test suite for $\Spec$ is both $\U$-complete and $\U'$-complete,
	then it is $\U \cup \U'$-complete.
\end{lemma}

A particular class of fault domains that has been widely studied is based on the maximal number of states of implementations.  

\begin{definition}
	Let $m>0$. Then fault domain $\U_m$ is the set of all 
	%complete 
	Mealy machines with at most $m$ states.
\end{definition}
	 
In the literature, $\U_m$-complete test suites are usually called \emph{$m$-complete}.
Suppose $m \geq n$, where $n$ is the number of states of a specification $\Spec$. 
Given that the size of $m$-complete test suites grows exponentially in $m-n$, a tester cannot possibly consider all Mealy machines with up to $m$ states, if $m-n$ is large. Therefore, we will propose alternative and much larger fault domains that can still be fully explored during testing.

We consider fault domains of Mealy machines in which all states can be reached by first performing an access sequence from a set $A$ (typically a state cover of specification $\Spec$), followed by $k$ arbitrary inputs, for some small $k$.  These fault domains capture the tester's assumption that when bugs in the implementation introduce extra states, these extra states can be reached via a few transitions from states reachable via scenarios from $A$.
The next definition formally introduces the corresponding fault domains $\U^A_k$.
A somewhat similar notion was previously proposed by Maarse \cite{Maarse20}, in a specific context with action refinement.

\begin{definition}
Let $k$ be a natural number and let $A \subseteq I^*$.
Then fault domain $\U^A_k$ is the set of all Mealy machines $\M$
such that every state of $\M$ can be reached by an input sequence $\sigma \rho$, for some $\sigma \in A$ and $\rho \in I^{\leq k}$.
\end{definition}

\begin{example}
	\label{ex:SPYH}
	Mealy machine $\M$ from Figure~\ref{Fig:turnstileCE} is contained in fault domain
	$\U_1^A$, for  $A = \{ \epsilon, c \}$, since all states of $\M$ can be reached via at most one transition from the two states $L'$ and $U'$ that are reachable via $A$.
	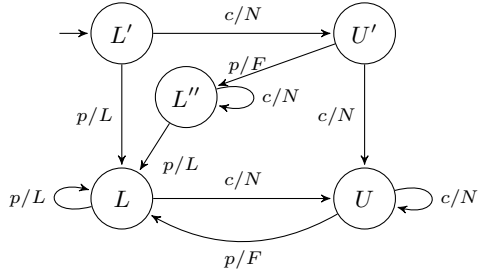
\begin{figure}[htb!]
		\begin{center}
			\begin{tikzpicture}[->,>=stealth',shorten >=1pt,auto,node distance=1.2cm,main node/.style={circle,draw,font=\sffamily\large\bfseries},
			]
			\def\yoffset{8mm}
			\node[initial,state,frontier] (L) {\treeNodeLabel{$L'$}};
			\node[state,frontier] (U) [right of=L,xshift=2cm] {\treeNodeLabel{$U'$}};
			\node[state,frontier] (Lp) [below of=L,yshift=-1cm] {\treeNodeLabel{$L$}};
			\node[state,frontier] (Lpp) [below right of=L] {\treeNodeLabel{$L''$}};
			\node[state,frontier] (Up) [below of=U,yshift=-1cm] {\treeNodeLabel{$U$}};
			
			\path[every node/.style={font=\sffamily\scriptsize}]
			(L) edge node {$c/N$} (U)
			edge node[left]{$p/L$} (Lp)
			(U) edge node[left] {$p/F$} (Lpp)
			edge node[left] {$c/N$} (Up)
			(Lp) edge node {$c/N$} (Up)
			edge [loop left] node {$p/L$} (Lp)
			(Up) edge [bend left] node {$p/F$} (Lp)
			edge [loop right] node[right] {$c/N$} (Up)
			(Lpp) edge [loop right] node {$c/N$} (Lpp)
			edge node {$p/L$} (Lp)
			;
			\end{tikzpicture}
			\caption{A Mealy machine $\M$ contained in fault domain $\U_1^A$, with $A = \{ \epsilon, c \}$.}
			\label{Fig:turnstileCE}
		\end{center}
	\end{figure}
\end{example}

\begin{remark}
The definition of $\U_k^A$ is closely related to the fundamental concept of \emph{eccentricity} known from graph theory \cite{DL2017}.
Consider a directed graph $G = (V, E)$.
For vertices $w, v \in V$, let $d(w,v)$ be the length of a shortest path from $w$ to $v$, or $\infty$ if no such a path exists.
The \emph{eccentricity} $\epsilon(w)$ of a vertex $w \in V$ is 
the maximum distance from $w$ to any other vertex:
\begin{eqnarray*}
	\epsilon(w) & = & \max_{v \in V} d(w,v).
\end{eqnarray*}
This definition generalizes naturally to subsets of vertices.
The distance from a set $W \subseteq V$ of vertices to a vertex $v \in V$, is the length of a shortest path from some vertex in $W$ to $v$, or $\infty$ if no such path exists.
The \emph{eccentricity} $\epsilon(W)$ of a set of vertices $W$ is
the maximum distance from $W$ to any other vertex:
\begin{eqnarray*}
	d(W, v)  =  \min_{w \in W}  d(w,v)
& ~~~~~~~~ &
\epsilon(W)  =  \max_{v \in V} d(W,v).
\end{eqnarray*}
We view Mealy machines as directed graphs in the obvious way.
In the example of Figure~\ref{Fig:turnstileCE}, the eccentricity of initial state $L'$ is $2$ (since every state of $\M$ can be reached with at most 2 transitions from $L'$) and the eccentricity of state $U'$ is $\infty$ (since there is no path from $U'$ to $L'$).
The eccentricity of $\{ L', U' \}$ is $1$, since state $L$ can be reached with a single transition from $L'$, and states $L''$ and $U$ can be reached with a single transition from $U'$.

Fault domain $\U^A_k$ can alternatively be defined as the set of Mealy machines $\M$ for which the eccentricity of the set of states reachable via $A$ is at most $k$. 
Note that, for a set of vertices $W$, $\epsilon(W)$ can be computed in linear time by contracting all elements of $W$ to a single vertex $w$, followed by a breadth-first search from $w$.
\end{remark}

The Mealy machine of Figure~\ref{Fig:turnstileCE}, which is contained in fault domain $\U_1^A$, has two states ($L'$ and $U'$) that are reached via a sequence from $A$, and three extra states that can be reached via a single transition from these two states.
More generally, if $A$ is a prefix closed set with $n$ sequences and the set of inputs $I$ contains $l$ elements, then at most $n$ states can be reached via sequence from $A$, and at most $nl - n + 1$ additional states can be reached via a single transition from states already reached by $A$.  A second step from $A$ may lead to $l (nl - n + 1)$ extra states, etc.
This leads us to the following proposition.
	
\begin{proposition}
	\label{number of states}
	Let $A \subset I^*$ be prefix closed with $|A| = n$, let $|I|=l$ and $k>0$. Then fault domain $\U_k^A$ contains Mealy machines with up to $n + (\sum_{j=0}^{k-1} l^j)(nl -n +1)$ states. 
\end{proposition}
	
This observation implies that the number of states of Mealy machines in $\U_k^A$ grows exponentially in $k$ when there are at least $2$ inputs. 
Even for small values of $k$, the number of states may increase appreciably.
Consider, for instance, the Mealy machine model for the Free BSD 10.2 TCP server that was obtained through black-box learning by Fiter\u{a}u-Bro\c{s}tean et al \cite{FJV16:long}. This model has 55 states and 13 inputs, so if $A$ is a minimal state cover, then fault domain $\U_2^A$ contains Mealy machines with up to 9309 states. 

\vspace{0.5em}
Even though the size of Mealy machines in $\U_k^A$ grows exponentially in $k$, fault domain $\U_m$ is not contained in fault domain $\U_k^A$ if $m = |A|+k$.  
For instance, the machine of Figure~\ref{Fig:turnstile} has two states and is therefore contained in $\U_2$.  However, this machine is not contained in $\U_0^A$ if we take $A = \{ \epsilon, p \}$.
We need to extend $\U_k^A$ in order to obtain a proper inclusion.
Suppose that $A$ is a minimal state cover for a minimal specification $\Spec$. Then states in $\Spec$ reached by sequences from $A$ will be pairwise inequivalent.
Methods for generating $m$-complete test suites typically first check whether states of implementation $\M$ reached by sequences from $A$ are also inequivalent. So these methods exclude any model $\M$ in which distinct sequences from $A$ reach equivalent states. This motivates the following definition:

\begin{definition}
	Let $A \subseteq I^*$. Then fault domain $\U^A$ is the set of all Mealy machines $\M$ such that there are  $\sigma, \rho \in A$ with $\sigma \neq \rho$ and $\delta^{\M}(q_0^{\M}, \sigma) \approx \delta^{\M}(q_0^{\M}, \rho)$.
\end{definition}

%If $A$ is a minimal state cover for a minimal specification $\Spec$ then (by construction) $\Spec$ is not an element of $\U^A$, and so $\U^A$ is not a fault domain for $\Spec$.
%However, $\U_k^A \cup \U^A$ is a fault domain.  
Note that $\U^A$ is infinite and contains Mealy machines with arbitrarily many states. 
Under reasonable assumptions, fault domain $\U_m$ is contained in fault domain $\U_k^A \cup \U^A$.

\begin{theorem}
	\label{theorem k-A larger than m}
	Let $A \subset I^*$ be a finite set of input sequences with $\epsilon \in A$. 
	Let $k$ and $m$ be natural numbers with $m = |A| + k$. 
	Then $\U_m \subseteq \U_k^A \cup \U^A$.
\end{theorem}

We refer to $\U_k^A \cup \U^A$-complete test suites as \emph{$k$-$A$-complete}.
By Theorem~\ref{theorem k-A larger than m} and Lemma~\ref{completeness inclusion}, any $k$-$A$-complete test suite is also $m$-complete, if $m = |A|+k$.
The converse of the inclusion of Theorem~\ref{theorem k-A larger than m} does not hold, as $\U^A$ may contain FSMs with an unbounded number of states.   
The next example illustrates that a test suite $T$ can be $m$-complete, but not $0$-$A$-complete for a state cover $A$ of the specification.

\begin{example}
The set $A = \{ \epsilon, a, aa \}$ is a minimal state cover of Mealy machine $\Spec$ at the left of Figure~\ref{Fig:UAkcomplete different from kAcomplete}.  Machine $\Spec$ is minimal and $aaa$ is a \emph{distinguishing sequence} since it generates different outputs for each of the three states of $\Spec$ .
Mealy machine $\Spec$ is trivially contained in fault domains $\U_3$ and $\U^A_0$,
whereas Mealy machine $\M$ at the right of Figure~\ref{Fig:UAkcomplete different from kAcomplete} is
clearly not contained in  $\U_3$ and $\U^A_0$.
\begin{figure}[ht!]
	\begin{center}
		%\vspace{-1cm}
		\begin{tikzpicture}[->,>=stealth',shorten >=1pt,auto,node distance=1.8cm,main node/.style={circle,draw,font=\sffamily\large\bfseries},
		]
		\def\yoffset{8mm}
		\node[initial,state,frontier] (0) {$s_0$};
		\node[state,frontier] (1) [right of=0] {$s_1$};
		\node[state,frontier] (2) [right of=1] {$s_2$};
		
		\node[initial,state,frontier] [right of=2,xshift=0.5cm](q0) {$q_0$};
		\node[state,frontier] (q1) [right of=q0] {$q_1$};
		\node[state,frontier] (q2) [right of=q1] {$q_2$};
		\node[state,frontier] (q3) [above of=q1] {$q_3$};
		
		\path[every node/.style={font=\sffamily\scriptsize}]
		(0) edge [bend left] node[above] {$a/0$} (1)
			edge [bend right] node[below] {$b/0$} (1)
		(1) edge [bend left] node[above] {$a/0$} (2)
			edge [bend right] node[below] {$b/0$} (2)
		(2) edge [loop above] node[above] {$a/1$} (2)
			edge [loop below] node[below] {$b/1$} (2)
		(q0) edge  node[below] {$b/0$} (q1)
			edge  node[sloped,above] {$a/0$} (q3)	
		(q1) edge node[above] {$a/0$} (q2)
			edge [bend right] node[below] {$b/0$} (q2)
		(q2) edge [loop above] node[above] {$a/1$} (q2)
			edge [loop below] node[below] {$b/1$} (q2)
		(q3) edge [loop right] node[right] {$a/0$} (q3)	
			edge node[sloped,above] {$b/0$} (q2)
			;
		\end{tikzpicture}
		\caption{A specification $\cal S$ (left) and an inequivalent implementation $\M$ (right).}
		\label{Fig:UAkcomplete different from kAcomplete}
	\end{center}
\end{figure}
However, $\M$ is contained in fault domain $\U^A$ since $\delta^{\M}(q_0, a) = \delta^{\M}(q_0, aa)$.
Note that the set $B = \{ \epsilon, b, bb \}$ is also a minimal state cover for $\Spec$, and the sequence $bbb$ is also a distinguishing sequence.
Using the Wp-method, to be discussed in more detail in Section~\ref{sec:applications}, $0$-$B$-complete test suites can easily be built from a state cover and a distinguishing sequence.  In particular, the set 
%\begin{eqnarray*}
$T  =  B \cdot \{ bbb \} ~~ \cup ~~ B \cdot \{ a, b \} \cdot \{ bbb \}$
%\end{eqnarray*}
is a $0$-$B$-complete test suite for $\Spec$, and therefore (by Theorem~\ref{theorem k-A larger than m}) also $3$-complete.
However, since $\M$ passes test suite $T$, $T$ is not $\U_0^A \cup \U^A$-complete, and thus not $0$-$A$-complete.		
\end{example}
%However, we have the following trivial inclusion for the base case where $k=0$.  
%
%\begin{proposition}
%	\label{special case k=0}
%	Let $A \subset I^*$ be a finite set of input sequences with $|A|=m$.  
%	Then $\U_0^A \subseteq \U_m$ (and thus $\U_0^A \cup \U^A \subseteq \U_m \cup \U^A$).
%\end{proposition}

Below we give an example to show that, for $k>0$, $m$-complete test suites generated by the SPYH-method \cite{SouchaB18} are not always $k$-$A$-complete, if $m = |A|+k$.
Appendix~\ref{sc: appendix}
contains variations of this example, which demonstrate that the SPY-method \cite{SPY12} and the H-method \cite{DorofeevaEY05} are not $k$-$A$-complete either.

\begin{example}
	\label{ex:SPYHcnt}
	Consider specification $\Spec$ and testing tree $\Obs$ from Figure~\ref{Fig:turnstile}. This specification and the corresponding test suite $T = \{ cccp, ccpp, cppp, pcpcp, ppp \}$ were both taken from \cite{SouchaB18}, where the SPYH-method was used to generate $T$, which was shown to be $3$-complete for $\Spec$. 
	Consider the minimal state cover $A = \{ \epsilon, c \}$ for $\Spec$.
    FSM $\M$ from Figure~\ref{Fig:turnstileCE} belongs to fault domain
    $\U_1^A$, since all states can be reached via at most one transition from $L'$ or $U'$.
    Clearly $\Spec \not\approx \M$, as input sequence $cpcp$ provides a counterexample.
	Nevertheless, $\M$ passes test suite $T$.
	Thus the test suite generated by the SPYH-method \cite{SouchaB18} is not $1$-$A$-complete.
\end{example}

\section{A Sufficient Condition for $k$-$A$-Completeness}
\label{sec: main result}

In this section, we describe a sufficient condition for a test suite to be $k$-$A$-complete, which (based on ideas of \cite{DorofeevaEY05,Vaandrager2024,VaandragerGRW22}) is phrased entirely in terms of properties of its testing tree. This tree should contain access sequences for each state in the specification, successors for these states for all possible inputs should be present up to depth $k+1$, and apartness relations between certain states of the tree should hold.
Before we present our condition and its correctness proof, we first need to introduce the concepts of %bisimulation, 
apartness, basis and stratification, and study their basic properties.

\subsection{Apartness}

In our sufficient condition, the concept of \emph{apartness}, inspired by a similar notion that is standard in constructive real analysis  \cite{troelstra_schwichtenberg_2000,GJapartness},  plays a central role.  

\begin{definition}[Apartness]
	For a Mealy machine $\M$, we say that states $q,r\in Q^{\M}$ are \emph{apart}
	(written $q \apart r$) iff there is some $\sigma\in I^*$ such that
	$\semantics{q}(\sigma)\converges$, $\semantics{r}(\sigma)\converges$,
	and $\semantics{q}(\sigma) \neq \semantics{r}(\sigma)$.
	We say that $\sigma$ is a \emph{separating sequence} for $q$ and $r$. We also call $\sigma$ a \emph{witness} of $q\apart r$ and write $\sigma
	\vdash q\apart p$.
\end{definition}
Note that the apartness relation $\mathord{\apart}\subseteq Q\times Q$ is irreflexive and symmetric. 
For the observation tree of Figure~\ref{Fig:turnstile} we may derive the following apartness pairs and corresponding witnesses:
$p \vdash 0 \apart 1$ and $p \vdash 0 \apart 11$.
Observe that when two states are apart they are not equivalent, but states that are not equivalent are not necessarily apart.  States $0$ and $12$, for instance, are neither equivalent nor apart.
However, for complete Mealy machines apartness coincides with inequivalence.

The apartness of states $q\apart r$ expresses that there is a conflict in their
semantics, and consequently, apart states can never be identified by a
functional simulation.

\begin{lemma}
	\label{la: apartness refinement}
	For a functional simulation $f\colon \Obs\to \M$,
	\[
	q\apart r\text{ in }\Obs
	\qquad\Longrightarrow
	\qquad
	f(q)\apart f(r)\text{ in }\M \qquad\text{for all }q, r\in Q^{\Obs}.
	\]
\end{lemma}
Thus, whenever states are apart in the observation tree $\Obs$, we
know that the corresponding states in Mealy machine $\M$ are distinct.
The apartness relation satisfies a weaker version of
\emph{co-transitivity}, stating that if $\sigma\vdash r\apart r'$ and $q$ has
the transitions for $\sigma$, then $q$ must be apart from at least one of $r$
and $r'$, or maybe even both.
\begin{lemma}[Weak co-transitivity]
	\label{la: weak co-transitivity}
	In every Mealy machine $\M$,
	\[
	\sigma \vdash r \apart r' ~\wedge~ \stateCan{q}{\sigma}  ~~\Longrightarrow~~
	r \apart q ~\vee~ r' \apart q
	\qquad\text{for all }r,r',q\in Q^{\M}, \sigma\in I^*.
	\]
\end{lemma}

%We will use the weak co-transitivity property during learning. For instance in \autoref{Fig:Obs}, by posing the output query $aba$, consisting of the access sequence for $t_1$ concatenated with the witness $ba$ for $t_0 \apart t_2$, the learner extends the observation tree and may infer that $t_1 \apart t_2$.
%A tester may use the weak co-transitivity property during testing. For instance in \autoref{Fig:Obs}, by performing the test $aba$, consisting of the access sequence for $t_1$ concatenated with the witness $ba$ for $t_0 \apart t_2$, 
%co-transitivity ensures that $t_0 \apart t_1$ or $t_2 \apart t_1$. By inspecting the outputs,
%a tester may conclude that $t_2 \apart t_1$.

\subsection{Basis}

In each observation tree, we may identify a \emph{basis}: an ancestor closed set of states that are pairwise apart.  In general, a basis is not uniquely determined, and an observation tree may for instance have different bases with the same size. However, once we have fixed a basis, the remaining states in the tree can be uniquely partitioned by looking at their distance from this basis.

\begin{definition}[Basis]
	Let $\Obs$ be an observation tree. A nonempty subset of states $B \subseteq Q^{\Obs}$ is called a \emph{basis} of $\Obs$ if
	\begin{enumerate}
		\item 
	 $B$ is ancestor-closed: for all $q \in B : q \neq q_0^{\Obs} ~ \Rightarrow ~ \mathsf{parent}(q) \in B$, and
	 \item
	 states in $B$ are pairwise apart: for all $q, q'\in B : q \neq q' ~ \Rightarrow ~ q \apart q'$.
\end{enumerate}
For each state $q$ of $\Obs$, the \emph{candidate set} $C(q)$ is the set of basis states that are not apart from $q$: $C(q) ~=~ \{ q'\in B \mid \neg ( q \apart q') \}$.
State $q$ is \emph{identified} if $| C(q) | = 1$.
\end{definition}

Since $B$ is nonempty and ancestor-closed, all states on the access path of a basis state are in the basis as well. In particular, the initial state $q^{\Obs}_0$ is in the basis. Also note that, by definition, basis states are identified.
The next lemma gives some useful properties of a basis.

\begin{lemma}
	\label{obs tree minimal state cover}
	Suppose $\Obs$ is an observation tree for $\M$ with $f \colon \Obs \to \M$ and basis $B$ such that  $| B | = | Q^{\M} |$. Then $f$ restricted to $B$ is a bijection, $\M$ is minimal, and $\mathsf{access}(B)$ is a minimal state cover for $\M$.
\end{lemma}

Whenever a subset of states $B$ of a testing tree is a basis, the
corresponding test suite is $\U^A$-complete, for $A = \mathsf{access}(B)$.

\begin{lemma} \label{La:basis}
	Let $\Spec$ be a Mealy machine, let $T$ be a test suite for $\Spec$, let $B$ be a basis for $\Obs = \mathsf{Tree}(\Spec, T)$, and let $A = \mathsf{access}(B)$.
	Then $T$ is $\U^A$-complete.
\end{lemma}

\subsection{Stratification}

A basis $B$ induces a stratification of observation tree $\Obs$: first we have
the set $F^0$ of immediate successors of basis states that are not basis states themselves, next the set $F^1$ of immediate successors of states in $F^0$, etc.
In general, $F^k$ contains all states that can be reached via a path of length $k+1$ from $B$.
%Suppose $\Spec$ is a specification with state cover $A$, that is, $A \subseteq I^*$ is a finite, prefix closed set such that for every state $q \in Q^{\Spec}$ there exists a $\sigma \in A$ with $\delta^{\Spec}(q_0^{\Spec}, \sigma) = q$.
\begin{definition}[Stratification]
Let  $\Obs$ be an observation tree with basis $B$.
Then $B$ induces a \emph{stratification} of $Q^\Obs$ as follows.  For $k \geq 0$,
\begin{eqnarray*}
	F^k & = & \{ q \in Q^{\Obs} \mid d(B,q) = k+1 \}.
\end{eqnarray*}
We call $F^k$ the \emph{$k$-level frontier} and
%\begin{enumerate}
%	\item 
%	We write $F^0$ for the set of immediate successors of basis states that are not basis states themselves:
%	$
%	F^0 ~:=~ \{ q' \in Q^{\Obs} \setminus B \mid \exists q \in B, i \in I : q' = \delta^{\Obs}(q, i) \}.
%	$
%	We refer to $F^0$ as the \emph{$0$-level frontier}.
%	\item 
%	For $k>0$, the \emph{$k$-level frontier} $F^k$ is the set of immediate successors of $k-1$-level frontier states: 
%	$
%	F^k ~:=~ \{ q' \in Q^{\Obs} \mid \exists q \in F^{k-1}, i \in I : q' = \delta^{\Obs}(q, i) \}.
%	$
%\end{enumerate}
%For $k$ a natural number, 
write $F^{<k} = \bigcup_{0 \leq i<k} F^i$ and $F^{\leq k} = \bigcup_{0 \leq i \leq k} F^i$.
\end{definition}

\begin{example}
Figure~\ref{stratification} shows the stratification for an observation tree for specification $\Spec$ from Figure~\ref{Fig:Conformance} induced by basis $B = \{ 0, 1, 8 \}$.  
Witness $a a$ shows that the three basis states are pairwise apart, and therefore identified.
States from sets $B$, $F^0$, $F^1$ and $F^2$ are marked with  different colors.
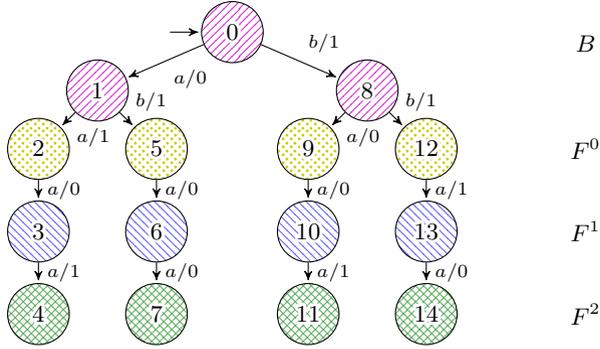
\begin{figure}[t!]
	\begin{center}
		\begin{tikzpicture}[->,>=stealth',shorten >=1pt,auto,node distance=1.3cm,main node/.style={circle,draw,font=\sffamily\large\bfseries}]
		\node[initial,state,basis] (0) {\treeNodeLabel{$0$}};
		\node[state,basis] (1) [below left of=0,xshift=-1cm] {\treeNodeLabel{$1$}};
		\node[state,frontier] (2) [below left of=1] {\treeNodeLabel{$2$}};
		\node[state,q1class] (3) [below of=2] {\treeNodeLabel{$3$}};
		\node[state,q2class] (4) [below of=3] {\treeNodeLabel{$4$}};
		\node[state,frontier] (5) [below right of=1,] {\treeNodeLabel{$5$}};
		\node[state,q1class] (6) [below of=5] {\treeNodeLabel{$6$}};
		\node[state,q2class] (7) [below of=6] {\treeNodeLabel{$7$}};
		\node[state,basis] (8) [below right of=0,xshift=1cm] {\treeNodeLabel{$8$}};
		\node[state,frontier] (9) [below left of=8] {\treeNodeLabel{$9$}};
		\node[state,q1class] (10) [below of=9] {\treeNodeLabel{$10$}};
		\node[state,q2class] (11) [below of=10] {\treeNodeLabel{$11$}};
		\node[state,frontier] (12) [below right of=8] {\treeNodeLabel{$12$}};		
		\node[state,q1class] (13) [below of=12] {\treeNodeLabel{$13$}};
		\node[state,q2class] (14) [below of=13] {\treeNodeLabel{$14$}};
		
		\node (F0) [right of=12,xshift=1cm] {$F^0$};
		\node (F1) [right of=13,xshift=1cm] {$F^1$};
		\node (F2) [right of=14,xshift=1cm] {$F^2$};
		\node (B) [above of=F0,yshift=0.3cm] {$B$};
		
		\path[every node/.style={font=\sffamily\scriptsize}]
		(0) edge node {$a/0$} (1)
			edge node {$b/1$} (8)
		(1) edge node {$a/1$} (2)
			edge node {$b/1$} (5)
		(8) edge node {$a/0$} (9)
			edge node {$b/1$} (12)
		(2) edge node {$a/0$} (3)
		(3) edge node {$a/1$} (4)
		(5) edge node {$a/0$} (6)
		(6) edge node {$a/0$} (7)
		(9) edge node {$a/0$} (10)
		(10) edge node {$a/1$} (11)
		(12) edge node {$a/1$} (13)
		(13) edge node {$a/0$} (14)
		;
		\end{tikzpicture}
	\end{center}
	\caption{Stratification of an observation tree induced by $B = \{ 0, 1, 8 \}$.}
	\label{stratification}
\end{figure}
In Figure~\ref{stratification}, $B$ is complete, but $F^0$, $F^1$ and $F^2$ are incomplete (since states of $F^0$ and $F^1$ have no outgoing $b$-transitions, and states of $F^2$ have no outgoing transitions at all).
The four $F^0$ states are also identified since $C(2) = \{ 0 \}$, $C(5) = \{ 8 \}$, $C(9) = \{ 0 \}$, and $C(12) = \{ 1 \}$.
Two states in $F^1$  are identified since $C(3) = C(10) = \{ 1 \}$, whereas the other two are not since $C(6) = C(13) = \{ 0, 8 \}$.
Since states in $F^2$ have no outgoing transitions, they are not apart from any other state, and thus $C(4) = C(7) = C(11) = C(14) = \{ 0, 1, 8 \}$.
\end{example}

\subsection{A Sufficient Condition for $k$-$A$-completeness}
We are now prepared to state our characterization theorem.  

\begin{theorem} \label{k-A-complete}
	Let $\M$ and $\Spec$ be Mealy machines, let $\Obs$ be an observation tree for both $\M$ and $\Spec$, let $B$ be a basis for $\Obs$ with $|B|=|Q^{\Spec}|$, let $A = \mathsf{access}(B)$, let $F^0, F^1,\ldots$ be the stratification induced by $B$, and let $k \geq 0$.
	Suppose $B$ and $F^{<k}$ are complete,
	all states in $F^k$ are identified, and the following condition holds:
	\begin{eqnarray}
	\label{cotransitivity requirement FV}
	\forall  q \in F^k ~ \forall r \in F^{< k} : &  & C(q) = C(r) \vee q \apart r
	\end{eqnarray}
	Suppose that $\M \in \U^A_k$.
	Then $\Spec\approx \M$.
\end{theorem}

As a corollary, we obtain a sufficient condition for $k$-$A$-completeness.

\begin{corollary} \label{Cy:k-A-complete}
	Let $\Spec$ be a Mealy machine, let $T$ be a test suite for $\Spec$, let $\Obs = \mathsf{Tree}(\Spec, T)$, let $B$ be a basis for $\Obs$ with $|B|=|Q^{\Spec}|$, let $A = \mathsf{access}(B)$, let $F^0, F^1,\ldots$ be the stratification of $\Obs$ induced by $B$, and let $k \geq 0$.
	Suppose $B$ and $F^{<k}$ are complete,
	all states in $F^k$ are identified, and condition (\ref{cotransitivity requirement FV}) holds.
	Then $T$ is $k$-$A$-complete.
\end{corollary}

The conditions of Corollary~\ref{Cy:k-A-complete} do not only impose restrictions on testing tree $\Obs$, but also on specification $\Spec$.
Suppose that the conditions of Corollary~\ref{Cy:k-A-complete} hold.
Then, by Lemma~\ref{lemma testing tree}, there is a functional simulation $f \colon \Obs \to \Spec$.
By Lemma~\ref{obs tree minimal state cover},
$f$ restricted to $B$ is a bijection, $\Spec$ is minimal, and $\mathsf{access}(B)$ is a minimal state cover for $\Spec$. Furthermore,
since $B$ is complete and $f$ restricted to $B$ is a bijection, $\Spec$ is also complete.
Minimality and completeness of specifications are common assumptions in conformance testing.

\begin{example}
A simple example of the application of Corollary~\ref{Cy:k-A-complete},
is provided by the observation tree from Figure~\ref{stratification} for the specification $\Spec$ from Figure~\ref{Fig:Conformance}. This observation tree corresponds to the test suite 
$T = \{ aaaa, abaa, baaa, bbaa \}$.
We claim that this test suite is $0$-$A$-complete, for $A = \{ \epsilon, a, b \} = \mathsf{access}(B)$.
Note that condition (\ref{cotransitivity requirement FV}) vacuously holds when $k=0$.  
All other conditions of Corollary~\ref{Cy:k-A-complete} are also met: basis $B$ is complete and all states in $F^0$ are identified.  Therefore, test suite $T$ is $0$-$A$-complete.
We may slightly optimize the test suite by replacing test $bb aa$ by test $b b a$, since all conditions of the corollary are still met for the reduced testing tree.
\end{example}

The conditions of Corollary~\ref{Cy:k-A-complete} are sufficient for $k$-$A$-completeness, but the following trivial example illustrates that they are not necessary.

\begin{example}
Consider the Mealy machines $\Spec$ of Figure~\ref{Fig:not necessary}, which has a single state, two inputs $a$ and $b$, and minimal state cover $A = \{ \epsilon \}$.
\begin{figure}[ht!]
	\begin{center}
		%\vspace{-1cm}
		\begin{tikzpicture}[->,>=stealth',shorten >=1pt,auto,node distance=2cm,main node/.style={circle,draw,font=\sffamily\large\bfseries},
		]
		\def\yoffset{8mm}
		\node[initial,state,frontier] (0) {};

		\path[every node/.style={font=\sffamily\scriptsize}]
		(0) edge [loop right] node[right] {$a/0$, $b/1$} (0)
			%edge [loop below] node[right] {$b/1$} (0)
	;
		\end{tikzpicture}
		\caption{Test suite $T = \{ ab \}$ is $0$-$\{ \epsilon \}$-complete for the above specification.}
		\label{Fig:not necessary}
	\end{center}
\end{figure}
Test suite $T = \{ ab \}$ does not meet the conditions of Corollary~\ref{Cy:k-A-complete}, since
the basis of the corresponding testing tree is not complete.
Nevertheless, we claim that $T$ is $0$-$A$-complete.
Since $A$ is a singleton, the fault domain $\U^A$ is empty.
The fault domain $\U^A_0$ only contains Mealy machines with a single state, and self-loop transitions for inputs $a$ and $b$.  Test suite $T$ verifies that the outputs of these transitions are in agreement with $\Spec$.  Thus, any machine in $\U^A_0$ that passes $T$ will be equivalent to $\Spec$. Hence $T$ is $0$-$A$-complete.
\end{example}

Theorem~\ref{k-A-complete} and Corollary~\ref{Cy:k-A-complete} do not require that all states in $F^{<k}$ are identified, but 
 %This requirement was included in \cite{Vaandrager2024}, but %B\'{a}lint Kocsis observed that it is
 % not used in the proof. 
%As it turns out, 
the other conditions of the theorem/corollary already imply this.
%that states in $F^{<k}$ are identified.
\begin{proposition} \label{full frontier identified}
	Let $\Obs$ be an observation tree for $\Spec$, $B$ a basis for $\Obs$ with $|B| = |Q^{\Spec}|$,  $F^0, F^1,\ldots$ the stratification induced by $B$, and $k \geq 0$.
	Suppose $B$ and $F^{<k}$ are complete, all states in $F^k$ are identified, and
	condition (\ref{cotransitivity requirement FV}) holds.
	Then all states in $F^{<k}$ are identified
\end{proposition}

Appendix~\ref{condition 1 is needed} 
contains an example to illustrate that the converse implication does not hold: even if all states in $B \cup F^{\leq k}$ are identified,
condition (\ref{cotransitivity requirement FV}) may not hold.
The next proposition asserts that
condition (\ref{cotransitivity requirement FV})  implies co-transitivity for a much larger collection of triples, namely triples of a basis state, a state in $F^i$ and a state in $F^j$, for all $i \neq j$.

\begin{proposition} \label{variant condition}
	Let $\Obs$ be an observation tree for $\Spec$, $B$ a basis for $\Obs$ with $|B| = |Q^{\Spec}|$,  $F^0, F^1,\ldots$ the stratification induced by $B$, and $k \geq 0$.
	Suppose $B$ and $F^{<k}$ are complete, all states in $F^k$ are identified, and
	condition (\ref{cotransitivity requirement FV}) holds.
	Then
	$
	\forall i, j ~ \forall  q \in F^i ~ \forall r \in F^j : 0 \leq i < j \leq k  ~~ \Rightarrow  ~~ C(q) = C(r) \vee q \apart r
	$
\end{proposition}

As a consequence of the next proposition, condition (\ref{cotransitivity requirement FV}) can be equivalently formulated as
\begin{eqnarray}
\label{cotransitivity requirement}
\forall  q \in F^k ~ \forall r \in F^{< k} ~ \forall s \in B : s \apart q &  \;\Rightarrow\;& s \apart r \vee q \apart r
\end{eqnarray}
Condition (\ref{cotransitivity requirement}) says that apartness is \emph{co-transitive} for triples of states in the observation tree consisting of a state in $F^k$, a state in $F^{<k}$, and a basis state.
Co-transitivity is a fundamental property of apartness \cite{troelstra_schwichtenberg_2000,GJapartness}.  

\begin{proposition} \label{candidate set cotransitivity}
	Let $\Obs$ be an observation tree for $\Spec$, $B$ a basis for $\Obs$, and $|B| = |Q^{\Spec}|$.
	Suppose $q$ and $r$ are states of $\Obs$ and $q$ is identified.
	Then
	\begin{eqnarray*}
		C(q) = C(r) \vee q \apart r & ~~ \Leftrightarrow ~~& [ \forall s \in B : s \apart q \; \Rightarrow \; s \apart r \vee q \apart r]
	\end{eqnarray*}
\end{proposition}

\subsection{Algorithm}
We will now present an algorithm that checks, for a given observation tree $\Obs$ with $N$ states,
in $\Theta(N^2)$ time, for all pairs of states, whether they are apart or not.
In practice, models of realistic protocols often have up to a few dozen states and inputs \cite{NeiderSVK97}. This means that if $k$ equals 2 or 3, the observation tree will contain up to a few thousand states, and so with some optimization (e.g.\  on the fly computation of apartness pairs) our algorithm is practical. For larger benchmarks the observation trees already contain millions of states, and our algorithm cannot be applied. Nevertheless, from a theoretical perspective it is interesting to know that the apartness relation (and hence also our sufficient condition for $k$-$A$-completeness) can be computed in polynomial time.
		
		\begin{algorithm}%[htb!]
			\small
			\begin{algorithmic}[1]
				\Function{FillApartnessArray}{$~$}
				\For{$q \in Q$}
				\For{$q'\in Q$}
				\If{$\neg \mathit{Visited}(q,q')$}
				{\sc ApartnessCheck}$(q,q')$
				\EndIf
				\EndFor
				\EndFor
				\State \Return $\mathit{Apart}$
				\EndFunction
				
				\Function{ApartnessCheck}{$q,q'$}		 
				\State $l \gets \mathit{Adj}[q]$, $l' \gets \mathit{Adj}[q']$
				\While{$l \neq \epsilon \wedge l' \neq \epsilon \wedge \neg \mathit{Apart}(q,q')$}
				\State $r \gets \mathsf{hd(l)}$, $r' \gets \mathsf{hd}(l')$
				\If{$in(r) < in(r')$}
				$l \gets \mathsf{tl}(l)$
				\ElsIf{$in(r') < in(r)$}
			 $l' \gets \mathsf{tl}(l')$
				\ElsIf{$out(r) = out(r')$}
				\If{$\neg \mathit{Visited}(r,r')$}
				 {\sc ApartnessCheck}$(r,r')$
				\EndIf
				\State $\mathit{Apart}(q,q') \gets \mathit{Apart}(q,q') \vee \mathit{Apart}(r,r')$
				\State $l \gets \mathsf{tl}(l)$, $l' \gets \mathsf{tl}(l')$
				\Else
				$\mathit{Apart}(q,q') \gets \mathit{true}$
				\EndIf
				\EndWhile
				\State $\mathit{Visited}(q,q') \gets \mathit{true}$
				\EndFunction
			\end{algorithmic}
%		}
		\caption{Computing the apartness relation}
		\label{algorithm apartness}
		\end{algorithm}
	
Algorithm~\ref{algorithm apartness} assumes a total order $<$ on the set of inputs $I$ and two partial functions $in : Q^{\Obs} \partialto I$ and $out : Q^{\Obs} \partialto O$ which, for each noninitial state $q$, specify the input and output, respectively, of the unique incoming transition of $q$.
It also assumes, for each state $q \in Q^{\Obs}$, an adjacency list $\mathit{Adj}[q] \in (Q^{\Obs})^*$ that contains the immediate successors of $q$, sorted on their inputs.
The algorithm maintains two Boolean arrays $\mathit{Visited}$ and $\mathit{Apart}$ to record
whether a pair of states $(q, q')$ has been visited or is apart, respectively.  
Initially, all entries in both arrays are $\mathit{false}$.
When exploring a pair of states $(q, q')$, Algorithm~\ref{algorithm apartness} searches for outgoing transitions of $q$ and $q'$ with the same input label.  In this case, if the outputs are different then it concludes
that $q$ and $q'$ are apart. Otherwise,
if the outputs are the same, it considers the target states $r$ and $r'$.  If the pair $(r, r')$ has not been visited yet then the algorithm recursively explores whether this pair of states is apart.
If $r$ and $r'$ are apart then $q$ and $q'$ are apart as well.	
	
The theorem below asserts the correctness and time complexity of Algorithm~\ref{algorithm apartness}.

\begin{theorem} \label{correctness apartness algorithm}
Algorithm~\ref{algorithm apartness} terminates with running time $\Theta(N^2)$, where $N = | Q^{\Obs} |$. 
Upon termination, $\mathit{Apart}(q,q') = \mathit{true}$ iff $q \apart q'$, for all $q, q'\in Q^{\Obs}$.
\end{theorem}

Once we know the apartness relation, there are several things we can do, e.g.,
we may check in $\Theta(N^2)$ time whether the conditions of Theorem~\ref{k-A-complete} hold
(but note that $N$ grows exponentially in $k$).
First we check in linear time whether all states in $B$ and $F^{<k}$ are complete. Next, in one pass over array $\mathit{Apart}$, we compute the candidate sets for all frontier states.  By Proposition~\ref{full frontier identified}, if some state in the frontier has a candidate set with more than one element, then the conditions of Theorem~\ref{k-A-complete} do not hold.  Otherwise, we can check in constant time whether frontier states have the same candidate set.  Finally, we check condition (\ref{cotransitivity requirement FV}) with a double for-loop over $F^k$ and $F^{<k}$.
% and a constant amount of work per case.

If we can compute the apartness relation, we may also select appropriate state identifiers on-the-fly in order to generate small $k$-$A$-complete test suites (similar to the H-method \cite{DorofeevaEY05}), and
we can check in $\Theta(N^2)$ time whether tests can be removed from a given test suite without compromising $k$-$A$-completeness.  

\section{Deriving Completeness for Existing Methods}
\label{sec:applications}

In this section, we show how $k$-$A$-completeness of two popular algorithms for test suite generation, the Wp and HSI methods, follows from Corollary~\ref{Cy:k-A-complete}.
We also present an alternative $m$-completeness proof of the H-method \cite{DorofeevaEY05}, which is a minor variation of the proof of Theorem~\ref{k-A-complete}.
In order to define the Wp and HSI methods, we need certain sets (of sets) of sequences.
The following definitions are based on \cite{ThesisJoshua}; see \cite{DorofeevaEMCY10,ThesisJoshua} for a detailed exposition. 

\begin{definition}
	Let $\Spec$ be a Mealy machine.
	\begin{itemize}
	%\item 
	%A \emph{characterization set} for $\Spec$ is a nonempty, finite set of input sequences $W$, such that for each pair of inequivalent states of $\Spec$, $W$ contains a separating sequence.
	\item 
	A \emph{state identifier} for a state $q \in Q^{\Spec}$ is a set $W_q \subseteq I^*$ such that for every inequivalent state $r \in Q^{\Spec}$, $W_q$ contains a separating sequence for $q$ and $r$, i.e., a witness for their apartness.
	\item 
	We write $\{ W_q \}_{q \in Q^{\Spec}}$ or simply $\{ W_q \}_q$ for a set that contains a state identifier $W_q$ for each $q \in Q^{\Spec}$. %We refer to such a set of sets as a \emph{family}.
	\item 
	If ${\cal W} = \{ W_q \}_q$ is a set of state identifiers, then the \emph{flattening} $\bigcup {\cal W}$ is the set $\{ \sigma \in I^* \mid \exists q \in Q^{\Spec} : \sigma \in W_q \}$.
	\item 
	If $W$ is a set of input sequences and ${\cal W} = \{ W_q \}_q$ is a set of state identifiers, then the \emph{concatenation} $W \odot {\cal W}$ is defined as
	$\{ \sigma \tau \mid \sigma \in W,~ \tau \in W_{\delta^{\Spec}(q_0^{\Spec}, \sigma)} \}$.
	\item 
	A set of state identifiers $\{ W_q \}_q$ is \emph{harmonized} if, for each pair of inequivalent states $q, r \in Q^{\Spec}$, $W_q \cap W_r$ contains a separating sequence for $q$ and $r$.
	We refer to such a set of state identifiers as a \emph{separating family}.
	\end{itemize}
\end{definition}

%For example, $W = \{ a, ba \}$ is a characterization set for the Mealy machine of Figure~\ref{Fig:Obs}(right), since $a$ is a separating sequence for $q_0$ and $q_2$, and for $q_1$ and $q_2$, and $ba$ is a separating sequence for $q_0$ and $q_1$.

$k$-$A$-completeness of the Wp-method \cite{FujiwaraEtAl91} follows from Corollary~\ref{Cy:k-A-complete} via routine checking. 

\begin{proposition}[$k$-$A$-completeness of the Wp-method]
	\label{completeness Wp}
	Let $\Spec$ be a complete, minimal Mealy machine, $k \geq 0$,
	$A$ a minimal state cover for $\cal S$, and ${\cal W} = \{ W_q \}_q$ a set of state identifiers. Then $T = (A \cdot I^{\leq k+1}) \cup (A \cdot I^{\leq k} \cdot \bigcup {\cal W}) \cup (A \cdot I^{\leq k+1} \odot {\cal W})$ is a $k$-$A$-complete test suite for $\Spec$.
\end{proposition}

Also $k$-$A$-completeness of
the HSI-method of Luo et al \cite{Luo1995} and Petrenko et al \cite{YP90,PetrenkoYLD93}
follows from Corollary~\ref{Cy:k-A-complete} via a similar argument.

\begin{proposition}[$k$-$A$-completeness of the HSI-method]
	\label{completeness HSI}
	Let $\cal S$ be a complete, minimal Mealy machine, $k \geq 0$,
	$A$ a minimal state cover for $\cal S$, and ${\cal W}$ a separating family. Then $T = (A \cdot  I^{\leq k+1}) \cup (A \cdot  I^{\leq k+1} \odot {\cal W})$ is $k$-$A$-complete for $\Spec$.
\end{proposition}

The $W$-method of \cite{Vas73,Ch78}, and the
UIOv-method of  \cite{ChanEtAl89} are instances of the Wp-method, 
and the ADS-method of \cite{LYa94} and the hybrid ADS method of \cite{SmeenkMVJ15} are instances of the HSI-method.
This means that $k$-$A$-completeness of these methods also follows from our results.

The H-method of Dorofeeva et al \cite{DorofeevaEY05} is based on a variant of our Theorem~\ref{k-A-complete} which requires that all states in $F^{\leq k}$ are identified and replaces condition (\ref{cotransitivity requirement FV}) by
\begin{eqnarray}
\label{cotransitivity requirement IM}
\forall  q, r \in F^{\leq k} : q \xrightarrow{+} r & \Rightarrow & C(q) = C(r) \vee q \apart r
\end{eqnarray}
By Proposition~\ref{variant condition}, condition (\ref{cotransitivity requirement IM}) is implied by condition (\ref{cotransitivity requirement FV}) of Theorem~\ref{k-A-complete}.
Appendix~\ref{sc: appendix} 
contains an example showing that the H-method is not $k$-$A$-complete. 
However, as shown by \cite[Theorem 1]{DorofeevaEY05}, the H-method is $m$-complete. 
Our condition (\ref{cotransitivity requirement FV}) can be viewed as a strengthening of condition (\ref{cotransitivity requirement IM}) of the H-method, needed for $k$-$A$-completeness.
Below we present an alternative formulation of the $m$-completeness result for the H-method of \cite{DorofeevaEY05}, restated for our setting. 
\iflong
Appendix~\ref{sc:proofs}
\else
The full version of this article
\fi
contains a proof of this proposition that uses the same proof technique as Theorem~\ref{k-A-complete}.

\begin{proposition}[$m$-completeness of the H-method] \label{k-complete}
	Let $\Spec$ be a Mealy machine with $n$ states, let $T$ be a test suite for $\Spec$, let $\Obs = \mathsf{Tree}(\Spec, T)$, let $B$ be a basis for $\Obs$ with $n$ states, let $A = \mathsf{access}(B)$, let $F^0, F^1,\ldots$ be the stratification of $\Obs$ induced by $B$, and let $m, k \geq 0$ with $m = n+k$.
	Suppose $B$ and $F^{<k}$ are complete,
	all states in $F^{\leq k}$ are identified, and condition (\ref{cotransitivity requirement IM}) holds.
%	\begin{eqnarray*}
%		\forall  q, r \in F^{\leq k} : q \xrightarrow{+} r & \Rightarrow & C(q) = C(r) \vee q \apart r
%	\end{eqnarray*}
	Then $T$ is $m$-complete.
\end{proposition}

\section{Conclusions and Future Work}
\label{sec: discussion}
In order to solve a long-standing open problem of Hennie \cite{Hennie64},
we proposed the notion of $k$-$A$-completeness.
%, which is based on a fault domain of FSMs in which all states can be reached by first performing a sequence from a state cover $A$ for the specification, followed by $k$ arbitrary inputs.  
We showed that the fault domain for $k$-$A$-completeness is larger than the one for $m$-completeness (if $m = |A|+ k$), and includes FSMs with a number of extra states that grows exponentially in $k$.
We provided a sufficient condition for $k$-$A$-completeness in terms of apartness of states of the observation induced by a test suite. Our condition can be checked efficiently (in terms of the size of the test suite) and can be used to prove $k$-$A$-completeness of 
the Wp-method of \cite{FujiwaraEtAl91} and the HSI-method of \cite{Luo1995,YP90,PetrenkoYLD93}.
Our results show that the Wp and HSI methods are complete for much larger fault domains than the ones for which they were originally designed.
We presented counterexamples to show that the SPY-method \cite{SPY12}, the H-method \cite{DorofeevaEY05}, and the SPYH-method \cite{SouchaB18} are not $k$-$A$-complete.

We view our results as an important step towards the definition of fault domains that capture realistic assumptions about possible faults, but still allow for the design of sufficiently small test suites that are complete for the fault domain and can be run within reasonable time.
A promising research direction is to reduce the size of the fault domains via reasonable assumptions on the structure of the implementation under test.  Kruger et al \cite{KrugerJR24} show that significantly smaller test suites suffice if the fault domains $\U_m$ are reduced by making plausible assumptions about the implementation.  The same approach can also be applied for the fault domains $\U_k^A \cup \U^A$ proposed in this article.
We expect that our results will find applications in the area of model learning
%Model learning, a.k.a.\ active automata learning, is emerging as a highly effective bug-funding technique with many applications 
\cite{Ang87,PeledVY99,Vaa17,HowarS2018}.
%For instance, using model learning, Fiter\u{a}u-Bro\c{s}tean e.a.\ \cite{FiterauEtAl17,FH17,FiterauEtal20,Fiterau-Brostean23} found numerous serious security vulnerabilities and non-conformance issues in implementations of several major network and security protocols including DTLS, SSH and TCP.
Black-box conformance testing, which is used to find counterexamples for hypothesis models, has become the main bottleneck in applications of model learning \cite{Vaa17,YangASLHCS19}.
This provides motivation and urgency to find fault domains that provide a better fit with applications and allow for smaller test suites \cite{KrugerJR24}.

An important question for future research is to explore empirical evidence for the usefulness of our new fault domains: how often are faulty implementations of real systems contained in fault domains $\U_k^A$, for small $k$ and a sensible choice of $A$?
H\"{u}bner et al \cite{HubnerHP19} investigated how test generation methods perform for SUTs whose behaviors lie outside the fault domain. They considered some realistic SUTs 
--- SystemC implementations of a ceiling speed monitor and an airbag controller described in SysML  ---
and used mutation operators to introduce bugs in a systematic and automated manner. Their experiments show that $m$-complete test suites generated by the  W- and Wp-methods exhibit significantly greater test strength than conventional random testing, even for behavior outside the fault domain.
It would be interesting to revisit these experiments and check which fraction of the detected faults is outside $\U_m$ but contained in $\U_k^A$.
%More generally, it would be interesting to investigate to what extent ``usual'' software/hardware bugs that occur in practice are contained in the fault model $\U_k^A$, for small values of $k$.

Our condition is sufficient but not necessary for completeness.
If one can prove, based on the assumptions of the selected fault domain, that two traces $\sigma$ and $\tau$ reach the same state both in specification $\Spec$ and in implementation $\M$, then it makes no difference in test suites whether a suffix $\rho$ is appended  after $\sigma$ or after $\tau$.
Sim\~{a}o et al \cite{SimaoPY09,SPY12} were the first to exploit this idea of \emph{convergence} to reduce the size of test suites in a setting for $m$-completeness.
It is future work to adapt these results to reduce the size of $k$-$A$-complete test suites.
Also, the complexity of deciding whether a test suite is $k$-$A$-complete is still unknown.
Another direction for future work is to lift our results to partial FSMs with observable nondeterminism, e.g.\  by adapting the state counting method \cite{PetrenkoYLD93}.

Closest to our characterization is the work of Sachtleben~\cite{Sachtleben24}, who develops unifying frameworks for proving $m$-completeness of test generation methods.
Inspired by the H-method, he defines an abstract H-condition that is sufficiently strong to prove $m$-completeness of the Wp, HSI, SPY, H and SPYH methods.
Sachtleben~\cite{Sachtleben24} also considers partially defined FSMs with observable nondeterminism, and takes convergence into account. Moreover, he mechanized all his proofs using Isabelle.
It would be interesting to explore whether the formalization of \cite{Sachtleben24} can be adapted to our notion of $k$-$A$-completeness.

An obvious direction for future work is to extend our notion of $k$-$A$-completeness to richer classes of models, such as timed automata and register automata.  The recent work of \cite{RK25} may provide some guidance here.
%Also close to our work is the result of Moerman \cite[Proposition 31, Chapter 2]{ThesisJoshua}, which provides  a sufficient condition for $m$-completeness of test suites. However, the condition of Moerman refers also to the implementation FSM, and therefore provides no efficient way to check $m$-completeness of a given test suite.
 %
%Our results suggest simple progress measures for performing a $k$-$A$-complete test suite, namely the sum of the number of elements of $B \cup F^{\leq k}$ and the number of established apartness pairs required for state identification and co-transitivity.
%This progress measure in a way quantifies our uncertainty about the correctness of the implementation.
%A natural question then is to search for test queries that lead to a maximal increase of the progress measure. In order to reduce our uncertainty as fast as possible and/or to find bugs as quickly as possible, it makes sense to give priority to these tests.
%
It will also be interesting to explore if our characterization can be used for efficient test suite generation, or for pruning test suites that have been generated by other methods.
%Of course, scalability may become an issue: for large specifications, large input alphabets and large values of $k$  it may no longer be feasible to store the full observation tree in main memory.  However, for many practical benchmarks (see e.g.\ \cite{NeiderSVK97}) it should not be a problem to handle the observation trees for $k$-$A$-complete test suites for $k=2$ or $k=3$.
%

%% Bibliography

\bibliographystyle{plainurl}
\bibliography{abbreviations,dbase}

\appendix
\newpage
\section{Condition (\ref{cotransitivity requirement FV}) is Needed}
\label{condition 1 is needed}

Proposition~\ref{full frontier identified} establishes that condition (\ref{cotransitivity requirement FV}) implies that all states in $F^{<k}$ are identified.
The converse implication does not hold: even if all states in $B \cup F^{\leq k}$ are identified, condition (\ref{cotransitivity requirement FV}) may not hold.
The Mealy machines $\cal S$ and $\M$ of Figure~\ref{Fig:co-transitivity needed} present a counterexample with $k=1$ and $A = \{ \epsilon, l, r \}$.
\begin{figure}[ht!]
	\begin{center}
		\begin{tikzpicture}[->,>=stealth',shorten >=1pt,auto,node distance=1.8cm,main node/.style={circle,draw,font=\sffamily\large\bfseries},
			]
			\def\yoffset{8mm}
			\node[initial,state,frontier] (0) {\treeNodeLabel{$q_0$}};
			\node[state,frontier] (1) [right of=0] {\treeNodeLabel{$q_1$}};
			\node[state,frontier] (2) [below of=1] {\treeNodeLabel{$q_2$}};
			
			\node[initial,state,frontier] [right of=1,xshift=2cm](q0) {\treeNodeLabel{$q_0$}};
			\node[state,frontier] (q1) [right of=q0] {\treeNodeLabel{$q_1$}};
			\node[state,frontier] (q2) [below of=q1] {\treeNodeLabel{$q_2$}};
			\node[state,frontier] (q3) [left of=q0] {\treeNodeLabel{$q_3$}};

			\path[every node/.style={font=\sffamily\scriptsize}]
			(0) edge [bend left] node[sloped,above] {$r/0$} (1)
			edge node[sloped,above] {$l/0$} (2)
			(1) edge node[below] {$l/1$} (0)
			edge [bend left] node[right] {$r/0$} (2)
			(2) edge [bend left] node[sloped,below] {$r/1$} (0)
			edge  node[left] {$l/0$} (1)
			(q0) edge node[sloped, above] {$l/0$} (q2)
			edge [bend left] node {$r/0$} (q1)	
			(q1) edge node[below] {$l/1$} (q0)
			edge [bend left] node {$r/0$} (q2)
			(q2) edge [bend left] node[sloped,below] {$r/1$} (q3)
			edge node[left] {$l/0$} (q1)
			(q3) edge[loop above] node {$l/0$} (q3)
			edge [bend left=45] node {$r/0$} (q1)
			;
		\end{tikzpicture}
		\caption{A specification $\cal S$ (left) and a faulty implementation $\M$ (right).}
		\label{Fig:co-transitivity needed}
	\end{center}
\end{figure}
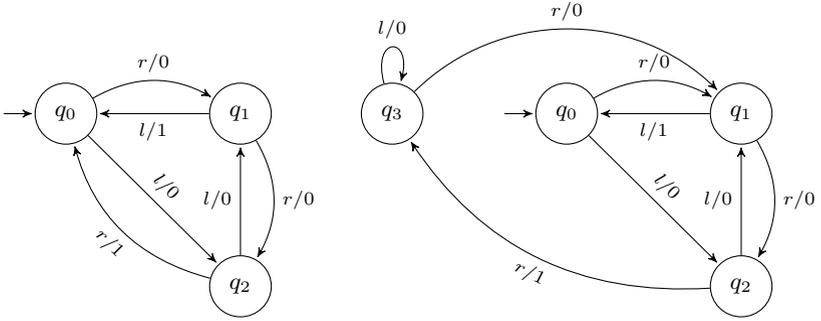
Note that these machines are not equivalent: input sequence $r r r l l l$ distinguishes them.
The extra state $q_3$ of $\M$ behaves similar as state $q_0$ of $\cal S$, but is not equivalent.  
Figure \ref{observation tree M and S} shows an observation tree $\Obs$ for both $\cal S$ and $\M$. 
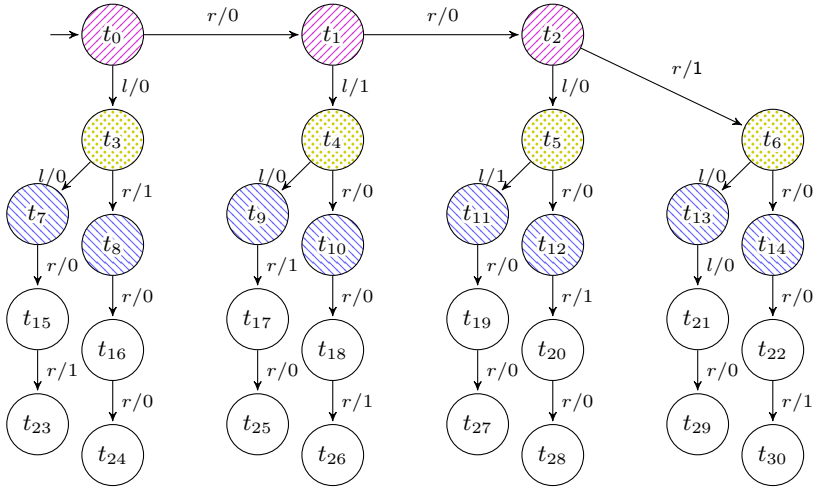
\begin{figure}[b!]
	\begin{center}
		\begin{tikzpicture}[->,>=stealth',shorten >=1pt,auto,node distance=1.4cm,main node/.style={circle,draw,font=\sffamily\large\bfseries}]
			\node[initial,state,basis] (0) {\treeNodeLabel{$t_0$}};
			\node[state,basis] (1) [right of=0,xshift=1.5cm] {\treeNodeLabel{$t_1$}};
			\node[state,basis] (2) [right of=1,xshift=1.5cm] {\treeNodeLabel{$t_2$}};
			\node[state,frontier] (3) [below of=0] {\treeNodeLabel{$t_3$}};
			\node[state,frontier] (4) [below of=1] {\treeNodeLabel{$t_4$}};
			\node[state,frontier] (5) [below of=2,] {\treeNodeLabel{$t_5$}};
			\node[state,frontier] (6) [right of=5,xshift=1.5cm] {\treeNodeLabel{$t_6$}};
			\node[state,q1class] (7) [below left of=3] {\treeNodeLabel{$t_7$}};
			\node[state,q1class] (8) [below of=3] {\treeNodeLabel{$t_8$}};
			\node[state,q1class] (9) [below left of=4] {\treeNodeLabel{$t_9$}};
			\node[state,q1class] (10) [below  of=4] {\treeNodeLabel{$t_{10}$}};
			\node[state,q1class] (11) [below left of=5] {\treeNodeLabel{$t_{11}$}};
			\node[state,q1class] (12) [below of=5] {\treeNodeLabel{$t_{12}$}};		
			\node[state,q1class] (13) [below left of=6] {\treeNodeLabel{$t_{13}$}};
			\node[state,q1class] (14) [below  of=6] {\treeNodeLabel{$t_{14}$}};
			\node[state] (15) [below of=7] {\treeNodeLabel{$t_{15}$}};
			\node[state] (16) [below of=8] {\treeNodeLabel{$t_{16}$}};
			\node[state] (17) [below of=9] {\treeNodeLabel{$t_{17}$}};
			\node[state] (18) [below of=10] {\treeNodeLabel{$t_{18}$}};
			\node[state] (19) [below of=11] {\treeNodeLabel{$t_{19}$}};
			\node[state] (20) [below of=12] {\treeNodeLabel{$t_{20}$}};
			\node[state] (21) [below of=13] {\treeNodeLabel{$t_{21}$}};
			\node[state] (22) [below of=14] {\treeNodeLabel{$t_{22}$}};
			\node[state] (23) [below of=15] {\treeNodeLabel{$t_{23}$}};
			\node[state] (24) [below of=16] {\treeNodeLabel{$t_{24}$}};
			\node[state] (25) [below of=17] {\treeNodeLabel{$t_{25}$}};
			\node[state] (26) [below of=18] {\treeNodeLabel{$t_{26}$}};
			\node[state] (27) [below of=19] {\treeNodeLabel{$t_{27}$}};
			\node[state] (28) [below of=20] {\treeNodeLabel{$t_{28}$}};
			\node[state] (29) [below of=21] {\treeNodeLabel{$t_{29}$}};
			\node[state] (30) [below of=22] {\treeNodeLabel{$t_{30}$}};
			
			\path[every node/.style={font=\sffamily\scriptsize}]
			(0) edge node {$r/0$} (1)
			edge node {$l/0$} (3)
			(1) edge node {$r/0$} (2)
			edge node {$l/1$} (4)
			(2) edge node {$r/$1} (6)
			edge node {$l/0$} (5)
			(3) edge node[left] {$l/0$} (7)
			edge node {$r/1$} (8)
			(4) edge node[left] {$l/0$} (9)
			edge node {$r/0$} (10)
			(5) edge node[left] {$l/1$} (11)
			edge node {$r/0$} (12)
			(6) edge node[left] {$l/0$} (13)
			edge node {$r/0$} (14)
			(7) edge node {$r/0$} (15)
			(15) edge node {$r/1$} (23)
			(8) edge node {$r/0$} (16)
			(16) edge node {$r/0$} (24)
			(9) edge node {$r/1$} (17)
			(17) edge node {$r/0$} (25)
			(10) edge node {$r/0$} (18)
			(18) edge node {$r/1$} (26)
			(11) edge node {$r/0$} (19)
			(19) edge node {$r/0$} (27)
			(12) edge node {$r/1$} (20)
			(20) edge node {$r/0$} (28)
			(13) edge node {$l/0$} (21)
			(21) edge node {$r/0$} (29)
			(14) edge node {$r/0$} (22)
			(22) edge node {$r/1$} (30)
			;
		\end{tikzpicture}
	\end{center}
	\caption{Observation tree for FSMs $\cal S$ and $\M$ from Figure~\ref{Fig:co-transitivity needed}.}
	\label{observation tree M and S}
\end{figure}
Observation tree $\Obs$ meets all the requirements of Theorem~\ref{k-A-complete}, except condition (\ref{cotransitivity requirement FV}).
One way to think of $\Obs$ is that $\M$ cherry picks distinguishing sequences from $\cal S$ to ensure that the $F^1$ states are identified by a sequence for which $\cal S$ and $\M$ agree.
Note that $B$ and $F^0$ are both complete, and all states in $B$ and $F^{\leq 1}$ are identified.
However, $\Obs$ does not satisfy condition (\ref{cotransitivity requirement FV}) as $t_{13}$ is not apart from $t_6$,  $C(t_6) = \{ t_0 \}$ and $C(t_{13}) = \{t_2 \}$.
The example shows that without condition (\ref{cotransitivity requirement FV}), Theorem~\ref{k-A-complete} doesn't hold.

\newpage

\section{The SPY and H-methods are not $k$-$A$-Complete}
\label{sc: appendix}

\begin{example}
	\label{ex:SPY}
	Figure~\ref{Fig:SPY}(left) shows the running example $\Spec$ from the article by Sim\~{a}o, Petrenko and Yevtushenko that introduces the SPY-method \cite{SPY12}. Using this method, a 3-complete test suite $\{ aaaa, baababba$, $bbabaa \}$ was derived in \cite{SPY12}.
	Consider the minimal state cover $A = \{ \epsilon, a \}$ for $\Spec$.
	Implementation $\M$ from Figure~\ref{Fig:SPY}(right) is contained in fault domain
	$\U_1^A$, since all states can be reached via at most one transition from $0'$ and $1'$.
	Clearly $\Spec \not\approx \M$, as input sequence $aab$ provides a counterexample.
	Nevertheless, $\M$ passes the derived test suite.
	Thus the test suite generated by the SPY-method \cite{SPY12} is not $1$-$A$-complete.
	\begin{figure}[ht!]
		\begin{center}
			\begin{tikzpicture}[->,>=stealth',shorten >=1pt,auto,node distance=2.5cm,main node/.style={circle,draw,font=\sffamily\large\bfseries},
			]
			\def\yoffset{8mm}
			\node[initial,state,frontier] (s0) {\treeNodeLabel{$s_0$}};
			\node[state,frontier] (s1) [right of=s0] {\treeNodeLabel{$s_1$}};
			\node[initial,state,frontier] (0p) [right of=s1,xshift=0.5cm] {\treeNodeLabel{$0'$}};
			\node[state,frontier] (1p) [right of=0p,xshift=2cm] {\treeNodeLabel{$1'$}};
			\node[state,frontier] (0) [below of=0p,yshift=-1cm] {\treeNodeLabel{$0$}};
			\node[state,frontier] (0pp) [below right of=0p] {\treeNodeLabel{$0''$}};
			\node[state,frontier] (1) [below of=1p,yshift=-1cm] {\treeNodeLabel{$1$}};
			
			\path[every node/.style={font=\sffamily\scriptsize}]
			(s0) edge node {$a/1$} (s1)
			edge [loop below] node {$b/0$} (s0)
			(s1) edge [bend left] node {$a/0$} (s0)
			edge [loop below] node {$b/0$} (s1)
			(0p) edge node {$a/1$} (1p)
			edge node[left]{$b/0$} (0)
			(1p) edge node[left] {$a/0$} (0pp)
			edge node[left] {$b/0$} (1)
			(0) edge node {$a/1$} (1)
			edge [loop below] node {$b/0$} (0)
			(1) edge [bend left] node {$a/0$} (0)
			edge [loop below] node {$b/0$} (1)
			(0pp) edge [loop left] node {$b/1$} (0pp)
			edge node {$a/1$} (1)
			;
			\end{tikzpicture}
			\caption{An implementation from fault domain
				$\U_1^A$ (right) that passes the $3$-complete test suite $\{ aaaa$, $baababba$, $bbabaa \}$ that was constructed for the specification (left) using the SPY-method.}
			\label{Fig:SPY}
		\end{center}
	\end{figure}
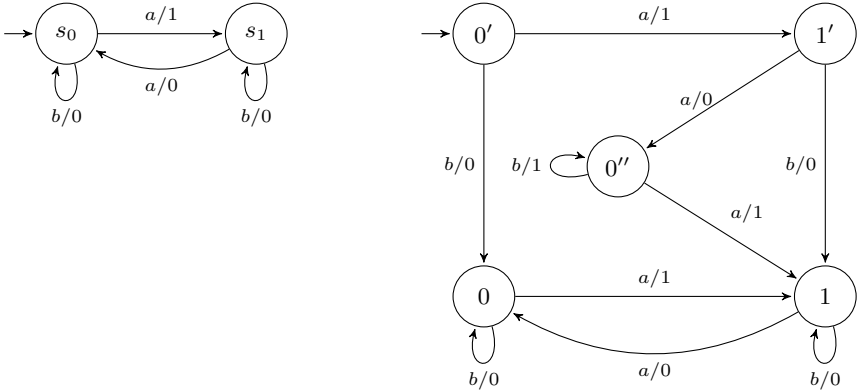
\end{example}

\begin{example}
		Figure~\ref{Fig:abc} shows the tree for a $3$-complete test suite generated by the H-method of Dorofeeva et al \cite{DorofeevaEY05} for the machine of Figure~\ref{Fig:H}(left). 
		This tree satisfies condition (\ref{cotransitivity requirement IM}) since the only transitions from $F^0$ to $F^1$ that change the candidate set are $a$-transitions, and the sources and targets of those transitions are apart.
		The machine of Figure~\ref{Fig:H}(right) will pass this test suite, even though the two machines are inequivalent ($cbc$ is a counterexample).
		It is easy to check that the machine on the right is in $\U_1^A$, for $A = \{ \epsilon, a \}$.
		Thus the test suite generated by the H-method is not $1$-$A$-complete.
		Indeed, the test suite does not meet condition (\ref{cotransitivity requirement FV}) since (for example) the states with access sequences $cb$ and $ac$ have different candidate sets but are not apart.
	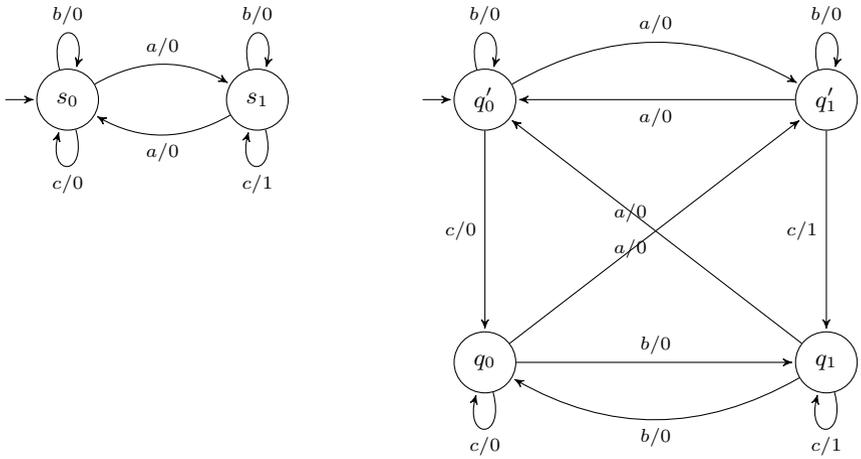
\begin{figure}[ht!]
		\begin{center}
			\begin{tikzpicture}[->,>=stealth',shorten >=1pt,auto,node distance=2.5cm,main node/.style={circle,draw,font=\sffamily\large\bfseries},
			]
			\def\yoffset{8mm}
			\node[initial,state,frontier] (s0) {\treeNodeLabel{$s_0$}};
			\node[state,frontier] (s1) [right of=s0] {\treeNodeLabel{$s_1$}};
			\node[initial,state,frontier] (0p) [right of=s1,xshift=0.5cm] {\treeNodeLabel{$q'_0$}};
			\node[state,frontier] (1p) [right of=0p,xshift=2cm] {\treeNodeLabel{$q'_1$}};
			\node[state,frontier] (0) [below of=0p,yshift=-1cm] {\treeNodeLabel{$q_0$}};
			\node[state,frontier] (1) [below of=1p,yshift=-1cm] {\treeNodeLabel{$q_1$}};
			
			\path[every node/.style={font=\sffamily\scriptsize}]
			(s0) edge[bend left] node {$a/0$} (s1)
				 edge [loop below] node {$c/0$} (s0)
				 edge [loop above] node {$b/0$} (s0)
			(s1) edge [bend left] node {$a/0$} (s0)
				 edge [loop below] node {$c/1$} (s1)
				 edge [loop above] node {$b/0$} (s1)
			(0p) edge[bend left] node {$a/0$} (1p)
				 edge node[left]{$c/0$} (0)
				 edge [loop above] node {$b/0$} (0p)
			(1p) edge node {$a/0$} (0p)
				 edge node[left] {$c/1$} (1)
				 edge [loop above] node {$b/0$} (1p)
			(0) edge node {$b/0$} (1)
				edge [loop below] node {$c/0$} (0)
				edge node{$a/0$} (1p)
			(1) edge [bend left] node {$b/0$} (0)
				edge [loop below] node {$c/1$} (1)
				edge node{$a/0$} (0p)
			;
			\end{tikzpicture}
			\caption{Specification (left) and implementation (right) from fault domain $\U_1^A$ that passes the test suite of Figure~\ref{Fig:abc}.}
			\label{Fig:H}
		\end{center}
	\end{figure}

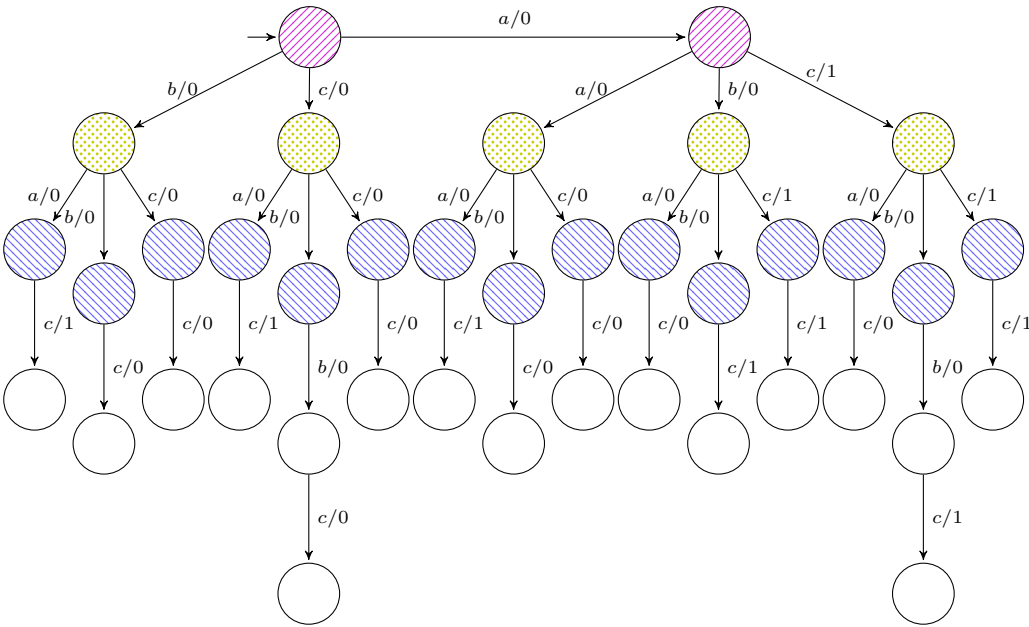
\begin{figure}[hbt!]
	\begin{center}
		\begin{tikzpicture}[->,>=stealth',shorten >=1pt,auto,node distance=2cm,main node/.style={circle,draw,font=\sffamily\large\bfseries}]
		\node[initial,state,basis] (init) {};
		\node[state,basis] (a) [right of=init,xshift=3.4cm] {};
		
		\node[state,frontier] (b) [below left of=init,xshift=-1.3cm] {};
		\node[state,frontier] (c) [right of=b,xshift=0.7cm] {};
		\node[state,frontier] (aa) [right of=c,xshift=0.7cm] {};
		\node[state,frontier] (ab) [right of=aa,xshift=0.7cm] {};
		\node[state,frontier] (ac) [right of=ab,xshift=0.7cm] {};
		
		\node[state,q1class] (ba) [below left of=b,xshift=0.5cm] {};
		\node[state,q1class] (bb) [below of=b] {};
		\node[state,q1class] (bc) [below right of=b,xshift=-0.5cm] {};
		\node[state,q1class] (ca) [below left of=c,xshift=0.5cm] {};
		\node[state,q1class] (cb) [below of=c] {};
		\node[state,q1class] (cc) [below right of=c,xshift=-0.5cm] {};
		\node[state,q1class] (aaa) [below left of=aa,xshift=0.5cm] {};
		\node[state,q1class] (aab) [below of=aa] {};
		\node[state,q1class] (aac) [below right of=aa,xshift=-0.5cm] {};
		\node[state,q1class] (aba) [below left of=ab,xshift=0.5cm] {};
		\node[state,q1class] (abb) [below of=ab] {};
		\node[state,q1class] (abc) [below right of=ab,xshift=-0.5cm] {};
		\node[state,q1class] (aca) [below left of=ac,xshift=0.5cm] {};
		\node[state,q1class] (acb) [below of=ac] {};
		\node[state,q1class] (acc) [below right of=ac,xshift=-0.5cm] {};
			
		\node[state] (bac) [below of=ba] {};
		\node[state] (bbc) [below of=bb] {};
		\node[state] (bcc) [below  of=bc] {};
		\node[state] (cac) [below  of=ca] {};
		\node[state] (cbb) [below of=cb] {};
		\node[state] (cbbc) [below of=cbb] {};
		\node[state] (ccc) [below  of=cc] {};
		\node[state] (aaac) [below  of=aaa] {};
		\node[state] (aabc) [below of=aab] {};
		\node[state] (aacc) [below of=aac] {};
		\node[state] (abac) [below  of=aba] {};
		\node[state] (abbc) [below of=abb] {};
		\node[state] (abcc) [below of=abc] {};
		\node[state] (acac) [below of=aca] {};
		\node[state] (acbb) [below of=acb] {};
		\node[state] (acbbc) [below of=acbb] {};
		\node[state] (accc) [below of=acc] {};
		
		\path[every node/.style={font=\sffamily\scriptsize}]
		(init)  edge node[above] {$a/0$} (a)
				edge node[left] {$b/0$} (b)
			    edge node[right] {$c/0$} (c)	
		(a) edge node[left] {$a/0$} (aa)
			edge node[right] {$b/0$} (ab)
			edge node[above] {$c/1$} (ac)	
		(b) edge node[left] {$a/0$} (ba)	
			edge node[left] {$b/0$} (bb)
			edge node[right] {$c/0$} (bc)	
		(c) edge node[left] {$a/0$} (ca)
			edge node[left] {$b/0$} (cb)
			edge node[right] {$c/0$} (cc)
		(aa) edge node[left] {$a/0$} (aaa)
			 edge node[left] {$b/0$} (aab)
			 edge node[right] {$c/0$} (aac)
		(ab) edge node[left] {$a/0$} (aba)
			 edge node[left] {$b/0$} (abb)
			 edge node[right] {$c/1$} (abc)
		(ac) edge node[left] {$a/0$} (aca)
			edge node[left] {$b/0$} (acb)
			edge node[right] {$c/1$} (acc)
		(ba) edge node[right] {$c/1$} (bac)
		(bb) edge node[right] {$c/0$} (bbc)
		(bc) edge node[right] {$c/0$} (bcc)
		(ca) edge node[right] {$c/1$} (cac)
		(cb) edge node[right] {$b/0$} (cbb)
		(cc) edge node[right] {$c/0$} (ccc)
		(aaa) edge node[right] {$c/1$} (aaac)
		(aab) edge node[right] {$c/0$} (aabc)
		(aac) edge node[right] {$c/0$} (aacc)
		(aba) edge node[right] {$c/0$} (abac)
		(abb) edge node[right] {$c/1$} (abbc)
		(abc) edge node[right] {$c/1$} (abcc)
		(aca) edge node[right] {$c/0$} (acac)
		(acb) edge node[right] {$b/0$} (acbb)
		(acc) edge node[right] {$c/1$} (accc)
		(cbb) edge node[right] {$c/0$} (cbbc)
		(acbb) edge node[right] {$c/1$} (acbbc)
		;
		\end{tikzpicture}
	\end{center}
	\caption{Testing tree for 3-complete test suite constructed for the specification of Figure~\ref{Fig:H}(left) using the H-method.}
	\label{Fig:abc}
\end{figure}
\end{example}

\newpage

\section{Proofs}
\label{sc:proofs}

We defined the lifted transition and output functions inductively for sequences of the form $i \sigma$. The next lemma basically says that we could have equivalently defined it inductively for sequences of the form $\sigma i$.

\begin{lemma}
	\label{la lifted functions}
	Let $\M$ be a Mealy machine, $q \in Q$, $i \in I$ and $\sigma \in I^*$. Then
	$\delta(q, \sigma i)  =  \delta(\delta(q, \sigma), i)$ and
	$\lambda(q, \sigma i) = \lambda(q, \sigma) \cdot \lambda(\delta(q, \sigma), i)$.
\end{lemma}
\begin{proof}
	We prove $\delta(q, \sigma i)  =  \delta(\delta(q, \sigma), i)$ by induction on the length of $\sigma$:
	\begin{itemize}
		\item 
		If $\sigma = \epsilon$ then
		\[
		\delta(q, \sigma i) = \delta(q, i) = \delta(\delta(q, \epsilon), i) = \delta(\delta(q, \sigma), i)
		\]
		\item 
		Induction step.  Assume $\sigma = j \rho$ with $j \in I$ and $\rho \in I^*$.  If $\delta(q, j)$ is undefined, then both
		$\delta(q, j \rho i)$ and $\delta(\delta(q, j \rho), i)$ are undefined, and the identity holds.  If $\delta(q, j)$ is defined then
		\[
		\delta(q, j \rho i) = \delta(\delta(q, j), \rho i) \overset{\text{ind. hyp.}}{=} \delta(\delta(\delta(q,j), \rho), i) = \delta(\delta(q, j \rho), i)
		\]
	\end{itemize}
	Similarly, we prove $\lambda(q, \sigma i) = \lambda(q, \sigma) \cdot \lambda(\delta(q, \sigma), i)$ by induction on the length of $\sigma$:
	\begin{itemize}
		\item 
		If $\sigma = \epsilon$ then
		\[
		\lambda(q, \sigma i) = \lambda(q, i) = \epsilon \cdot  \lambda(q, i) = \lambda(q, \epsilon) \cdot \lambda (\delta(q, \epsilon), i) = \lambda(q, \sigma) \cdot \lambda(\delta(q, \sigma), i)
		\]
		\item 
		Induction step.  Assume $\sigma = j \rho$ with $j \in I$ and $\rho \in I^*$.  If $\delta(q, j)$ and $\lambda(q, j)$ are undefined, then both $\lambda(q, \sigma i)$ and $\lambda(q, \sigma) \cdot \lambda(\delta(q, \sigma), i)$ are undefined, and the identity holds.  If $\delta(q, j)$ and $\lambda(q, j)$ are defined then
		\begin{eqnarray*}
			\lambda(q, \sigma i) & = &  \lambda(q, j \rho i) = \lambda(q, j) \cdot \lambda(\delta(q, j), \rho i) 
			\\
			& \overset{\text{ind. hyp.}}{=}	& \lambda(q, j) \cdot \lambda(\delta(q, j), \rho) \cdot \lambda(\delta(\delta(q,j), \rho), i)\\
			& = & \lambda(q, j \rho) \cdot \lambda(\delta(q, j \rho), i)\\
			& = &  \lambda(q, \sigma) \cdot \lambda(\delta(q, \sigma), i)  
		\end{eqnarray*}
	\end{itemize}
\end{proof}

The next lemma, which follows via a simple inductive argument from the definitions, states that a functional simulation preserves the transition and output functions on words.

\begin{lemma}
	\label{la:refinement lifted to sigma}
	Suppose $f \colon \M \to \N$.
	Suppose $q, q'\in Q^{\M}$, $\sigma \in I^*$ and $\rho \in O^*$.
	Then $q\xrightarrow{\sigma/\rho}q'$ implies $f(q) \xrightarrow{\sigma/\rho} f(q')$.
\end{lemma}
\begin{proof}
	Suppose $q\xrightarrow{\sigma/\rho}q'$.
	Then $\lambda^{\M}(q,\sigma) = \rho$ and $\delta^{\M}(q, \sigma) = q'$.
	In order to establish $f(q) \xrightarrow{\sigma/\rho} f(q')$,
	we need to show that
	$\lambda^{\N}(f(q),\sigma) = \rho$ and $\delta^{\N}(f(q), \sigma) = f(q')$.
	We prove both equalities by induction on the length of $\sigma$.
	\begin{itemize}
		\item 
		Assume $\sigma = \epsilon$. Then 
		\[
		\lambda^{\N}(f(q),\sigma) = \lambda^{\N}(f(q),\epsilon) = \epsilon 
		= \lambda^{\M}(q,\epsilon) = \lambda^{\M}(q,\sigma) = \rho,
		\]
		\[
		\delta^{\N}(f(q), \sigma) = \delta^{\N}(f(q), \epsilon) = f(q) = f(\delta^{\M}(q, \epsilon)) = f(\delta^{\M}(q, \sigma)) = f(q').
		\]
		\item 
		Assume $\sigma = i \sigma'$. We have
		$\delta^{\M}(q, i \sigma') = \delta^{\M}(\delta^{\M}(q,i), \sigma')$ and
		$\lambda^{\M}(q, i \sigma') = \lambda^{\M}(q,i) \lambda^{\M}(\delta^{\M}(q,i), \sigma')$.
		Let $q'' = \delta^{\M}(q, i)$, $o = \lambda^{\M}(q,i)$ and
		$\rho'= \lambda^{\M} (q'', \sigma')$.
		Then we obtain $\rho = o \rho'$, $q\xrightarrow{i/o}q''$ and $q'' \xrightarrow{\sigma'/\rho'}q'$.
		Since $f$ is a functional simulation and by the induction hypothesis,
		$f(q)\xrightarrow{i/o} f(q'')$ and $f(q'') \xrightarrow{\sigma'/\rho'} f(q')$. Then
		\[
		\lambda^{\N}(f(q),\sigma) = \lambda^{\N}(f(q),i) \lambda^{\N}(\delta^{\N}(f(q),i), \sigma') = o \lambda^{\N}(f(q''), \sigma') = o \rho'= \rho
		\]
		\[
		\delta^{\N}(f(q), \sigma) = \delta^{\N}(\delta^{\N}(f(q), i), \sigma') = \delta^{\N}(f(q''), \sigma') = f(q')
		\]
		as required.
	\end{itemize}
\end{proof}

\begin{proofappendix}{lemma testing tree}
	We check the three conditions of Definition~\ref{def refinement}:
	\begin{enumerate}
		\item 
		$f(q_0^{\Obs}) = f(\epsilon) = \delta^{\Spec} (q_0^{\Spec}, \epsilon) = q_0^{\Spec}$.
		\item 
		Assume $\sigma \in Q^{\Obs}$ and $i \in I$ such that $\delta^{\Obs}(\sigma,i) \converges$. Then by Lemma~\ref{la lifted functions}:
		\[
		f(\delta^{\Obs}(\sigma,i)) = f(\sigma i) = \delta^{\Spec}(q_0^{\Spec}, \sigma i) =
		\delta^{\Spec}(\delta^{\Spec}(q_0^{\Spec}, \sigma), i) =  \delta^{\Spec}(f(\sigma),i).
		\] 
		\item 
		Assume $\sigma \in Q^{\Obs}$ and $i \in I$ such that $\lambda^{\Obs}(\sigma,i) \converges$. Then 
		\[
		\lambda^{\Obs}(\sigma,i) 
		= \lambda^{\Spec}( \delta^{\Spec} (q_0^{\Spec}, \sigma), i)
		= \lambda^{\Spec}(f(\sigma), i).
		\]
	\end{enumerate}
\end{proofappendix}

\begin{proofappendix}{observation tree when test suite passes}
	We prove both implications:
	\begin{itemize}
		\item 
		Suppose $\M$ passes $T$.
		We prove $f \colon \Obs\to \M$ by checking the three conditions of Definition~\ref{def refinement}:
		\begin{enumerate}
			\item 
			$f(q_0^{\Obs}) = f(\epsilon) = \delta^{\M}(q_0^{\M}, \epsilon) = q_0^{\M}$.
			\item 
			Assume $\delta^{\Obs}(\sigma, i) \converges$, for some $\sigma \in Q^{\Obs}$ and $i \in I$. Then by Lemma~\ref{la lifted functions}:
			\[
			f(\delta^{\Obs}(\sigma, i)) = f( \sigma i) = 
			\delta^{\M} (q_0^{\M}, \sigma i) = \delta^{\M} (\delta^{\M} (q_0^{\M}, \sigma), i) = \delta^{\M} (f(\sigma), i)
			\]
			\item 
			As $\M$ passes $T$, for all $\sigma \in T$, 
			$\lambda^{\M} (q_0^{\M}, \sigma) = \lambda^{\Spec} (q_0^{\Spec}, \sigma)$.
			This implies that, for all $\sigma i \in \mathit{Pref}(T)$,
			$\lambda^{\M} (q_0^{\M}, \sigma i) = \lambda^{\Spec} (q_0^{\Spec}, \sigma i)$.
			By Lemma~\ref{la lifted functions}, this implies
			\begin{eqnarray*}
				\lambda^{\M} (\delta^{\M}(q_0^{\M}, \sigma), i) & = & \lambda^{\Spec} (\delta^{\Spec}(q_0^{\Spec}, \sigma), i).
			\end{eqnarray*}
			Now assume $\lambda^{\Obs}(\sigma, i) \converges$, for some $\sigma \in Q^{\Obs}$ and $i \in I$.
			Then $\sigma i \in \mathit{Pref}(T)$, and therefore, by the above equation,
			\[
			\lambda^{\Obs}(\sigma, i) = \lambda^{\Spec}( \delta^{\Spec} (q_0^{\Spec}, \sigma), i) = \lambda^{\M} (\delta^{\M}(q_0^{\M}, \sigma), i)  = \lambda^{\M}(f(\sigma), i).
			\]
		\end{enumerate}
		\item 
		Suppose $f \colon \Obs \to \M$.
		Let $\sigma \in T$.  By construction, $\sigma$ is also a state of $\Obs$ and, for some $\rho$,
		$\epsilon \xrightarrow{\sigma/\rho} \sigma$.
		By Lemma~\ref{lemma testing tree} and Lemma~\ref{la:refinement lifted to sigma},
		$q_0^{\Spec} \xrightarrow{\sigma/\rho} \delta^{\Spec}(q_0^{\Spec}, \sigma)$.
		By the assumption and Lemma~\ref{la:refinement lifted to sigma},
		$q_0^{\M} \xrightarrow{\sigma/\rho} \delta^{\M}(q_0^{\M}, \sigma)$.
		Hence $\lambda^{\M} (q_0^{\M}, \sigma) = \rho = \lambda^{\Spec} (q_0^{\Spec}, \sigma)$.
	\end{itemize}	
\end{proofappendix}

\begin{proofappendix}{theorem k-A larger than m}
	Suppose $\M \in \U_m \setminus \U^A$. We must show that $\M \in \U_k^A$.
	%Since $\M \in \U_m$, we know that $\M$ is a complete Mealy machine.
	Since $\M \not\in \U^A$, the states of $\M$ reached by sequences from $A$ are pairwise inequivalent.
	Therefore, the set $B$ of states of $\M$ that can be reached by sequences from $A$ contains $|A|$ elements. Also, since $\epsilon \in A$ we know that $q_0^{\M} \in B$.
	Let $q$ be a state of $\M$.
	Since we only consider Mealy machines that are initially connected, there exists a sequence of inputs $\sigma$ that reaches $q$.  Now pick $\sigma$ in such a way that the number of states outside $B$ visited by $\sigma$ is minimal. Observe that each state outside $B$ is visited at most once by $\sigma$ (otherwise there would be a loop and the number of visited states outside $B$ would not be minimal).
	Since $\M$ has at most $m$ states and $B$ contains $|A|$ states, this means that $\sigma$ visits at most $k$ states outside $B$.
	Since at least one state in $B$ is visited (namely $q_0^{\M}$), this means that $q$ can be reached with a sequence $\rho \in I^{\leq k}$ from a state in $B$.
	This implies that $\M \in \U_k^A$, as required.
\end{proofappendix}

\paragraph*{Bisimulations}
In the proof of our completeness result, we use the concept of bisimulation relations and the well-known fact that for (deterministic) Mealy machines bisimulation equivalence coincides with the equivalence $\approx$ defined above.

\begin{definition}[Bisimulation]
	\label{def bisimulation MM}
	A \emph{bisimulation} between Mealy machines $\M$ and $\N$ is a relation $R
	\subseteq Q^{\M} \times Q^{\N}$ satisfying, for all $q \in Q^{\M}$, $r \in Q^{\N}$, and $i \in I$,
	\begin{enumerate}
		\item
		$q^{\M}_0 \; R \; q^{\N}_0$,
		\item 
		$q \mathrel{R} r$ implies $\delta^{\M}(q, i) \converges ~ \Leftrightarrow ~ \delta^{\N}(r, i) \converges$,
		\item
		$q \mathrel{R} r \wedge \delta^{\M}(q,i) \converges$ implies
		$\delta^{\M}(q,i) \mathrel{R} \delta^{\N}(r, i)$ and
		$\lambda^{\M}(q,i) = \lambda^{\N}(r, i)$.
	\end{enumerate}
	We write $\M \simeq \N$ if there exists a bisimulation relation between $\M$ and $\N$.
\end{definition}

%\begin{lemma}
%	\label{La:approx}
%	Given complete Mealy machines $\M$ and $\N$, the equivalence relation
%	$\mathord{\approx} \subseteq Q^{\M}\times Q^{\N}$ is a bisimulation.
%\end{lemma}

The next lemma, a variation of the classical result of \cite{Pa81}, is easy to prove.
\begin{lemma}
	\label{la:bisimulation}
	Let $\M, \N$ be Mealy machines.
	Then $\M \simeq \N$ iff $\M \approx \N$.
\end{lemma}

The next lemmas give some useful properties of a basis.

\begin{lemma}
	\label{f restricted to B injective}
	Suppose $\Obs$ is an observation tree for $\M$ with $f \colon \Obs \to \M$ and basis $B$. Then $f$ restricted to $B$ is injective.
\end{lemma}
\begin{proof}
	Let $q$ and $q'$ be two distinct states in $B$.  Since $q \apart q'$, we may conclude by Lemma~\ref{la: apartness refinement} that $f(q) \apart  f(q')$. Thus in particular $f(q) \neq f(q')$ and so $f$ restricted to $B$ is injective.
\end{proof}

\begin{proofappendix}{obs tree minimal state cover}
	%Pick a state $q$ of $\M$.  In order to prove that $A = \mathsf{access}(B)$ is a minimal state cover for $\M$, it suffices to show that $q$ is reached by a unique sequence in $A$.
	By Lemma~\ref{f restricted to B injective}, $f$ restricted to $B$ is injective.
	Since  $| B | = | Q^{\M} |$, we may conclude that $f$ is a bijection between $B$ and $Q^{\M}$.
	Since states in $B$ are pairwise apart, it follows by Lemma~\ref{la: apartness refinement} that states from $Q^{\M}$ are pairwise apart.
	This means that $\M$ is minimal.
	Since every state in $B$ is reached by a unique sequence in $\mathsf{access}(B)$, and $f$ is a bijection, we may use Lemma~\ref{la:refinement lifted to sigma} to conclude that also every state in $Q^{\M}$ is reached by a unique sequence in $\mathsf{access}(B)$.
	%Since $f$ is bijection, there exists a state $r \in B$ with $f(r)=q$.
	%Let $\sigma = \mathsf{access}(r)$.
	%Then $q_0^{\Obs} \xrightarrow{\sigma}_{\Obs} r$.
	%Therefore, by Lemma~\ref{la:refinement lifted to sigma}, $q_0^{\M} \xrightarrow{\sigma}_{\M} q$.
	%So $q$ is reached by a unique sequence in $A$, as required.
\end{proofappendix}

\begin{lemma}
	\label{candidate set}
	Suppose $\Obs$ is an observation tree for $\M$ with $f \colon \Obs \to \M$ and basis $B$ such that $| B | = | Q^{\M} |$. Let $q$ be a state of $\Obs$.
	Then there exists a state $r \in B$ with $r \in C(q)$ and $f(q)=f(r)$.
\end{lemma}
\begin{proof}
	Let $f(q) = u$.
	By Lemma~\ref{obs tree minimal state cover}, $f$ restricted to $B$ is a bijection.
	Let $r \in B$ be the unique state with $f(r) = u$.
	Since $f(q) = f(r)$, Lemma~\ref{la: apartness refinement} implies that $q$ and $r$ are not apart.
	Hence $r \in C(q)$.
\end{proof}

\begin{proofappendix}{La:basis}
	Assume that $\M$ is a Mealy machine that passes $T$.
	By Lemma~\ref{observation tree when test suite passes}, $\Obs$ is an observation tree for $\M$.
	Let $f : \Obs \to \M$.
	Suppose $\sigma, \rho \in A$ with $\sigma \neq \rho$.
	Since $A = \mathsf{access}(B)$,
	$q = \delta^{\Obs}(q_0^{\Obs}, \sigma) \in B$ and $p = \delta^{\Obs}(q_0^{\Obs}, \rho) \in B$.
	Since $\Obs$ is a tree, $p \neq q$ and thus, as $B$ is a basis for $\Obs$, $q \apart p$.
	By Lemma~\ref{la: apartness refinement}, $f(q) \apart f(p)$.
	By Lemma~\ref{la:refinement lifted to sigma}, $f(q) = \delta^{\M}(q_0^{\M}, \sigma)$ and $f(p) = \delta^{\M}(q_0^{\M}, \rho)$. From this  we infer
	$\delta^{\M}(q_0^{\M}, \sigma) \not\approx \delta^{\M}(q_0^{\M}, \rho)$.
	This implies that distinct sequences from $A$ may never reach equivalent states in $\M$, and
	thus $\M \not\in \U^A$.
	Hence there are no Mealy machines in $\U^A$ that pass $T$, and we conclude that $T$ is $\U^A$-complete.
\end{proofappendix}

\begin{proofappendix}{k-A-complete}
	Let $f : \Obs \to \Spec$ and $g : \Obs \to \M$.  Define relation $R \subseteq Q^{\Spec} \times Q^{\M}$ by
	\begin{eqnarray*}
		(s, q) \in R & ~~ \Leftrightarrow ~~ &  [\exists t \in B \cup  F^{<k} : f(t) = s \wedge  g(t) = q].
	\end{eqnarray*}
	We claim that $R$ is a bisimulation between $\Spec$ and $\M$, as defined in Definition~\ref{def bisimulation MM}.
	\begin{enumerate}
		\item 
		Since $f$ is a functional simulation from $\Obs$ to $\Spec$, $f(q^{\Obs}_0) = q^{\Spec}_0$, and
		since $g$ is a functional simulation from $\Obs$ to $\M$, $g(q^{\Obs}_0) = q^{\M}_0$.
		Using $q^{\Obs}_0 \in B$, this implies $(q^{\Spec}_0, q^{\M}_0) \in R$.
		\item 
		Suppose $(s,q) \in R$ and $i \in I$. 
		Since $(s, q) \in R$, there is a $t \in B \cup F^{<k}$ such that $f(t) = s$ and $g(t) = q$.
		Since $B$ and $F^{<k}$ are complete, $\delta^{\Obs}(t,i) \downarrow$.
		Since $f$ and $g$ are functional simulations, also $\delta^{\Spec}(s,i) \downarrow$ and $\delta^{\M}(q,i) \downarrow$.
		Let $s' = \delta^{\Spec}(s,i)$, $q' = \delta^{\M}(q,i)$ and $t'= \delta^{\Obs}(t,i)$.
		Since $f$ and $g$ are functional simulations,
		$\lambda^{\Obs}(t,i) = \lambda^{\Spec}(s,i)$ and $\lambda^{\Obs}(t,i) = \lambda^{\M}(q,i)$.
		This implies $\lambda^{\Spec}(s,i) = \lambda^{\M}(q,i)$, as required. 
		Since $f$ and $g$ are functional simulations, $f(t') = s'$ and $g(t') = q'$.
		In order to prove $(s', q') \in R$, we consider two cases:
		\begin{enumerate}
			\item 
			$t' \in B \cup F^{<k}$.  In this case, since $f(t') = s'$ and $g(t') = q'$, $(s', q') \in R$ follows from the definition of $R$.
			\item 
			$t'\in F^k$. In this case, since we assume $\M \in \U^A_k$, there are sequences $\sigma \in \mathsf{access}(B)$ and $\rho \in I^{\leq k}$ such that $q'$ is reached by $\sigma \rho$.
			By the assumption that $B$ and $F^{<k}$ are complete, $t'' = \delta^{\Obs}(q_0^{\Obs}, \sigma\rho)$ is defined.
			By Lemma~\ref{la:refinement lifted to sigma}, $g(t'') = q'$.  
			Then, by Lemma~\ref{la: apartness refinement}, $t'$ and $t''$ are not apart.
			We claim that $t'$ and $t''$ have the same candidate set:  
			\begin{enumerate}
				\item 
				$t'' \in B$.   By definition of basis $B$, $C(t'') = \{ t'' \}$.
				Since not $t' \apart t''$ and $t'$ is identified, $C(t') = \{ t'' \}$.
				Hence $C(t'') = C(t')$.
				\item 
				$t'' \in F^{<k}$. Then by condition (\ref{cotransitivity requirement FV}) and since not $t' \apart t''$, $C(t'') = C(t')$.
			\end{enumerate}
			Since $t'$ is identified, $C(t') = \{ r \}$, for some $r \in B$.
			By Lemma~\ref{candidate set}, $f(t')=f(r)$.
			Since $C(t'') = C(t')$, also $C(t'') = \{ r \}$. Applying Lemma~\ref{candidate set} again gives
			$f(t'') = f(r)$. Hence $f(t'') = f(t') = s'$. This in turn implies $(s', q') \in R$, which completes the proof that $R$ is bisimulation.
		\end{enumerate}
	\end{enumerate}
	The theorem now follows by application of Lemma~\ref{la:bisimulation}.
\end{proofappendix}

\begin{proofappendix}{Cy:k-A-complete}
	Let $\M$ be a Mealy machine in $\U_k^A$.
	%Since $T$ is a test suite for $\Spec$, clearly $\M$ will pass $T$ if $\M \approx \Spec$.
	Assume that $\M$ passes $T$.
	By Lemma~\ref{observation tree when test suite passes}, $\Obs$ is an observation tree for both $\M$ and $\Spec$.
	Now Theorem~\ref{k-A-complete} implies $\M \approx \Spec$, and thus $T$ is $\U_k^A$-complete.
	By Lemma~\ref{La:basis}, $T$ is also $\U^A$-complete.
	Now $k$-$A$-completeness follows by Lemma~\ref{completeness union}.
\end{proofappendix}

\begin{proofappendix}{full frontier identified}
	Proof by contradiction.
	Assume that some state $q \in F^{<k}$ is not identified.
	Then there are distinct states $r, s \in B$ such that $\{ r, s \} \subseteq C(q)$.
	By definition of a basis, states $r$ and $s$ are apart.
	Let $\sigma$ witness the apartness of $r$ and $s$.
	Using that $F^{<k}$ is complete, we rerun $\sigma$ from state $q$ until we reach $F^k$.
	If $\delta^{\Obs}(q, \sigma) \in F^{\leq k}$, then by the
	weak co-transitivity Lemma~\ref{la: weak co-transitivity}, $q$ is apart from $r$ or from $s$, and we have a contradiction.
	Otherwise, let $\rho$ be the proper prefix of $\sigma$ with $\delta^{\Obs}(q, \rho) \in F^k$.
	If $\rho \vdash r \apart s$ then, by Lemma~\ref{la: weak co-transitivity}, $q$ is either apart from $r$ or from $s$, and we have a contradiction.
	So we may assume $\rho \not\vdash r \apart s$.
	Let $\delta^{\Obs}(q, \rho) = q'$, $\delta^{\Obs}(r, \rho) = r'$ and $\delta^{\Obs}(s, \rho) = s'$.
	Since $\sigma \vdash r \apart s$ but $\rho \not\vdash r \apart s$, we conclude $r'\apart s'$.
	Observe that both $r'$ and $s'$ are contained in $B \cup F^{<k}$.
	If both $r'$ and $s'$ are in $B$ then, since $q'$ is identified, $q'$ is apart from $r'$ or from $s'$.
	If $q' \apart r'$ then $q \apart r$ and we have a contradiction.
	If $q' \apart s'$ then $q \apart s$ and we have a contradiction.
	Otherwise, w.l.o.g., assume $r'\in F^{<k}$.
	Then, by condition (\ref{cotransitivity requirement FV}),
	$C(q') = C(r')$ or $q'\apart r'$.
	If $q' \apart r'$ then $q \apart r$ and we have a contradiction.
	So we conclude $C(q') = C(r')$.
	Since $q'$ is identified $C(q') = \{ t \}$, for some $t \in B$.
	If $s'\in B$ then, since $C(r') = \{ t \}$, $s'\neq t$ because otherwise $r'$ and $s'$ are not apart.
	This implies $q' \apart s'$, which implies $q \apart s$, which is a contradiction.
	If $s'\in F^{<k}$ then, by condition (\ref{cotransitivity requirement FV}),
	$C(q') = C(s')$ or $q'\apart s'$.
	If $q' \apart s'$ then $q \apart s$ and we have a contradiction.
	So we conclude $C(q') = C(s') = \{ t \}$.
	Let $f \colon \Obs \to \Spec$.
	By Lemma~\ref{candidate set}, $f(r') = f(t) = f(s')$.
	But since $r'\apart s'$, Lemma~\ref{la: apartness refinement} gives $f(r') \neq f(s')$.
	Contradiction.	
\end{proofappendix}

\begin{proofappendix}{variant condition}
	Note that, by the previous Proposition~\ref{full frontier identified}, all states in $F^{\leq k}$ are identified.
	Assume $i$ and $j$ are indices with $0 \leq i < j \leq k$, $q \in F^i$, $r \in F^j$ and $C(q) \neq C(r)$.  It suffices to prove that $q \apart r$.
	Let $f \colon \Obs \to \Spec$.
	Since both $q$ and $r$ are identified and $C(q) \neq C(r)$, it follows by Lemma~\ref{candidate set} and Lemma~\ref{la: apartness refinement} that $f(q) \apart f(r)$.
	Let $\sigma$ be a separating sequence for $f(q)$ and $f(r)$.
	If $\sigma \vdash q \apart r$, we are done.  Otherwise, since  $F^{<k}$ is complete, there is a prefix $\rho$ of $\sigma$ with $\delta^{\Obs}(r, \rho) \in F^k$.
	Let $q' \in F^{<k}$ and $r'\in F^k$ be the unique states such that $q \xrightarrow{\rho} q'$ and
	$r \xrightarrow{\rho} r'$.
	If $\rho \vdash q \apart r$, we are done.
	Otherwise, we conclude $f(q') \apart f(r')$.
	By Lemma~\ref{candidate set} and since both $q'$ and $r'$ are identified, we conclude $C(q') \neq C(r')$.
	Therefore, by condition (\ref{cotransitivity requirement FV}), $q' \apart r'$.
	Let $\tau$ be a witness for the apartness of $q'$ and $r'$.
	Then $\rho \tau$ is a witness for the apartness of $q$ and $r$.
	Thus $q \apart r$, as required.	
\end{proofappendix}

\begin{proofappendix}{candidate set cotransitivity}
	%$\mbox{ }$
	\begin{itemize}
		\item 
		``$\Rightarrow$''
		Assume $C(q) = C(r) \vee q \apart r$.
		Suppose $s \in B$ with $s \apart q$.  We need to show $s \apart r \vee q \apart r$.
		By our assumption, if $q \apart r$ then we are done.
		So suppose $C(q) = C(r)$. Then, since $s \apart q$, $s \not\in C(q)$.  Therefore $s \not\in C(r)$, which implies $s \apart r$, as required.
		\item 
		``$\Leftarrow$''
		Assume $\forall s \in B : s \apart q \; \Rightarrow \; s \apart r \vee q \apart r$.
		Suppose not $q \apart r$.  We need to show $C(q) = C(r)$.
		Because $q$ is identified, all basis states except one are apart from $q$.
		Let $q'$ be the unique basis state that is not apart from $q$.
		By our assumption, $r$ is apart from all states in $B \setminus \{ q'\}$.
		Thus $C(r) \subseteq \{ q'\}$.
		By Lemma~\ref{candidate set}, $C(r)$ contains at least one state.
		Therefore, we conclude that $C(r) = \{ q'\}$.
		This implies $C(q)=C(r)$, as required.
	\end{itemize}	
\end{proofappendix}

\begin{proofappendix}{correctness apartness algorithm}
	Correctness follows since two states $q$ and $q'$ are apart if and only if either (1) 
	both have an outgoing transition for the same input but with a different output, or (2) 
	both have an outgoing transition for the same input, leading to states $r$ and $r'$, respectively, such that $r$ and $r'$ are apart.
	This is exactly what the algorithm checks.
	
	Function {\sc ApartnessCheck} is called exactly once for each pair of states $(q,q')$.
	The overall complexity of the algorithm is $\Theta(N^2)$, because $\Obs$ has $N-1$ transitions, and each pair of transitions $(q, r)$ and $(q', r')$ is considered at most once by the algorithm
	(during execution of {\sc ApartnessCheck}$(q, q')$). The amount of work for each pair of transitions
	$(q, r)$ and $(q', r')$ is constant.
\end{proofappendix}

\begin{proofappendix}{completeness Wp}
	Since $\Spec$ is complete, $T$ is a test suite for $\Spec$.
	Let $\Obs = \mathsf{Tree}(\Spec, T)$ and $f : \Obs \to \Spec$.
	Let $B$ be the subset of states of $\Obs$ reached via an access sequence in $A$, and
	let $F^0, F^1,\ldots$ be the stratification induced by $B$.
	We check that the assumptions of Corollary~\ref{Cy:k-A-complete} hold:
	\begin{enumerate}
		\item 
		$B$ is a basis: 
		Since $A$ is a state cover for $\Spec$, it is prefix-closed. Hence, set $B$ is ancestor-closed.
		Suppose $q$ and $q'$ are two distinct states in $B$.
		We show that $q \apart q'$.
		Let $\sigma$ and $\sigma'$ be the access sequences of $q$ and $q'$, respectively.
		Then $\sigma \neq \sigma'$.
		Let $r = \delta^{\Spec}(q_0^{\Spec}, \sigma)$ and  $r' = \delta^{\Spec}(q_0^{\Spec}, \sigma')$.
		By Lemma~\ref{la:refinement lifted to sigma}, $f(q)=r$ and $f(q') = r'$.
		Since $A$ is a minimal state cover, $r \neq r'$, and since $\Spec$ is minimal, $r \not\approx r'$.
		Set $\bigcup {\cal W}$ contains a separating sequence $\rho$ for $r$ and $r'$.
		Since $A \cdot \bigcup {\cal W} \subseteq T$, $\delta^{\Obs}(q, \rho) \downarrow$ and $\delta^{\Obs}(q', \rho) \downarrow$.
		By Lemma~\ref{la:refinement lifted to sigma},
		$\lambda^{\Obs}(q, \rho) = \lambda^{\Spec}(r,\rho) \neq \lambda^{\Spec}(r',\rho) = \lambda^{\Obs}(q', \rho)$.
		Thus $\rho \vdash q \apart q'$, as required.
		\item
		Since $A \subseteq T$, $|A| = |B|$, and since $A$ is a minimal state cover, $|A| = | Q^{\Spec} |$.
		Hence $|B| = | Q^{\Spec} |$.
		\item 
		By construction, $A = \mathsf{access}(B)$.
		\item
		Since $A \cdot I^{\leq k+1} \subseteq T$, sets $B$ and $F^{<k}$ are complete.
		\item 
		All states in $F^k$ are identified:
		We show a stronger statement, namely that all states in $F^{\leq k}$ are identified.
		Suppose $r \in F^{\leq k}$.
		Let $\mathsf{access}(r) = \sigma$ and $s = \delta^{\Spec}(q_0^{\Spec}, \sigma)$.
		By Lemma~\ref{la:refinement lifted to sigma}, $f(r) = s$.
		By Lemma~\ref{obs tree minimal state cover}, $f$ restricted to $B$ is a bijection. Let $q \in B$ be the unique state with $f(q)=s$.
		Now suppose $q' \in B$ is distinct from $q$. Let $f(q') = s'$.
		By definition of a state identifier, $W_s$ contains a separating sequence $\rho$ for $s$ and $s'$.
		Since $A \cdot I^{\leq k+1} \odot {\cal W} \subseteq T$, $\delta^{\Obs}(r, \rho) \downarrow$.
		Since $ A \cdot \bigcup {\cal W} \subseteq T$,  $\delta^{\Obs}(q', \rho) \downarrow$.
		By Lemma~\ref{la:refinement lifted to sigma},
		\[
		\lambda^{\Obs}(q', \rho) = \lambda^{\Spec}(s',\rho) \neq \lambda^{\Spec}(s,\rho) = \lambda^{\Obs}(r, \rho).
		\]
		Thus $\rho \vdash r \apart q'$.
		Since $q'$ was chosen to be an arbitrary basis state different from $q$, this implies that $r$ is identified.
		\item 
		Condition (\ref{cotransitivity requirement FV}) holds:
		Suppose $q \in F^k$ and $r \in F^{<k}$.
		By the previous item, both $q$ and $r$ are identified, that is, there exist $q', r' \in B$ such that $C(q) = \{ q' \}$ and $C(r) = \{r'\}$.
		If $q'= r'$ then $C(q)=C(r)$ and we are done.  So assume $q'\neq r'$.
		Let $s = f(q')$ and $t = f(r')$. 
		By Lemma~\ref{candidate set}, $f(q) = s$ and $f(r)=t$.
		Since $f$ restricted to $B$ is a bijection, $s \neq t$,
		and since $\Spec$ is minimal, $s \not\approx t$.
		By definition of a state identifier, $W_s$ contains a separating sequence $\rho$ for $s$ and $t$.
		Since $A \cdot I^{\leq k+1} \odot {\cal W} \subseteq T$, $\delta^{\Obs}(q, \rho) \downarrow$.
		Since $A \cdot I^{\leq k} \cdot \bigcup {\cal W} \subseteq T$, $\delta^{\Obs}(r, \rho) \downarrow$.
		By Lemma~\ref{la:refinement lifted to sigma},
		$\lambda^{\Obs}(q, \rho) = \lambda^{\Spec}(s,\rho) \neq \lambda^{\Spec}(t,\rho) = \lambda^{\Obs}(r, \rho)$.
		Thus $\rho \vdash r \apart q$.
	\end{enumerate} 
	Since all conditions of Corollary~\ref{Cy:k-A-complete} hold, we conclude that $T$ is $k$-$A$-complete.
\end{proofappendix}

\begin{proofappendix}{completeness HSI}
	Since $\Spec$ is complete, $T$ is a test suite for $\Spec$.
	Let $\Obs = \mathsf{Tree}(\Spec, T)$ and $f : \Obs \to \Spec$.
	Let $B$ be the subset of states of $\Obs$ reached via an access sequence in $A$, and
	let $F^0, F^1,\ldots$ be the stratification induced by $B$.
	We check that the assumptions of Corollary~\ref{Cy:k-A-complete} hold:
	\begin{enumerate}
		\item 
		$B$ is a basis: 
		Since $A$ is a state cover for $\Spec$, it is prefix-closed. Hence set $B$ is ancestor-closed.
		Suppose $q$ and $q'$ are two distinct states in $B$.
		We show that $q \apart q'$.
		Let $\sigma$ and $\sigma'$ be the access sequences of $q$ and $q'$, respectively.
		Then $\sigma \neq \sigma'$.
		Let $r = \delta^{\Spec}(q_0^{\Spec}, \sigma)$ and  $r' = \delta^{\Spec}(q_0^{\Spec}, \sigma')$.
		By Lemma~\ref{la:refinement lifted to sigma}, $f(q)=r$ and $f(q') = r'$.
		Since $A$ is a minimal state cover, $r \neq r'$, and since $\Spec$ is minimal, $r \not\approx r'$.
		Since ${\cal W} = \{ W_q \}_q$ is a separating family, $W_r \cap W_{r'}$ contains a separating sequence $\rho$ for $r$ and $r'$.
		Since $A \odot {\cal W} \subseteq T$, $\delta^{\Obs}(q, \rho) \downarrow$ and $\delta^{\Obs}(q', \rho) \downarrow$.
		By Lemma~\ref{la:refinement lifted to sigma},
		$\lambda^{\Obs}(q, \rho) = \lambda^{\Spec}(r,\rho) \neq \lambda^{\Spec}(r',\rho) = \lambda^{\Obs}(q', \rho)$.
		Thus $\rho \vdash q \apart q'$, as required.
		\item
		Since $A \subseteq T$, $|A| = |B|$, and since $A$ is a minimal state cover, $|A| = | Q^{\Spec} |$.
		Hence $|B| = | Q^{\Spec} |$.
		\item 
		By construction, $A = \mathsf{access}(B)$.
		\item
		Since $A \cdot  I^{\leq k+1} \subseteq T$, sets $B$ and $F^{<k}$ are complete.
		\item 
		All states in $F^k$ are identified:
		We show a stronger statement, namely that all states in $F^{\leq k}$ are identified.
		Suppose $r \in F^{\leq k}$.
		Let $\mathsf{access}(r) = \sigma$ and let $s = \delta^{\Spec}(q_0^{\Spec}, \sigma)$.
		By Lemma~\ref{la:refinement lifted to sigma}, $f(r) = s$.
		By Lemma~\ref{obs tree minimal state cover}, $f$ restricted to $B$ is a bijection. 
		Let $q \in B$ be the unique state with $f(q)=s$.
		Now suppose $q' \in B$ is distinct from $q$, and let $f(q') = s'$.
		Since $f$ is a bijection, $s \neq s'$, and since $\Spec$ is minimal $s \not\approx s'$.
		By the definition of a separating family, $W_s \cap W_{s'}$ contains a separating sequence $\rho$ for $s$ and $s'$.
		Since $A \cdot  I^{\leq k+1} \odot {\cal W} \subseteq T$, $\delta^{\Obs}(r, \rho) \downarrow$ and $\delta^{\Obs}(q', \rho) \downarrow$.
		By Lemma~\ref{la:refinement lifted to sigma},
		$\lambda^{\Obs}(q', \rho) = \lambda^{\Spec}(s',\rho) \neq \lambda^{\Spec}(s,\rho) = \lambda^{\Obs}(r, \rho)$.
		Thus $\rho \vdash r \apart q'$.
		Since $q'$ was chosen to be an arbitrary basis state different from $q$, this implies that state $r$ is identified.
		\item 
		Condition (\ref{cotransitivity requirement FV}) holds:
		Suppose $q \in F^k$ and $r \in F^{<k}$.
		By the previous item, both $q$ and $r$ are identified, that is, there exist $q', r' \in B$ such that $C(q) = \{ q' \}$ and $C(r) = \{r'\}$.
		If $q'= r'$ then $C(q)=C(r)$ and we are done.  So assume $q'\neq r'$.
		Let $s = f(q')$ and $t = f(r')$. 
		By Lemma~\ref{candidate set}, $f(q) = s$ and $f(r)=t$.
		Since $f$ restricted to $B$ is a bijection, $s \neq t$, and
		since $\Spec$ is minimal, $s \not\approx t$.
		By the definition of a separating family, $W_s \cap W_t$ contains a separating sequence $\rho$ for $s$ and $t$.
		Since $A \cdot  I^{\leq k+1} \odot {\cal W} \subseteq T$, $\delta^{\Obs}(q, \rho) \downarrow$ and $\delta^{\Obs}(r, \rho) \downarrow$.
		By Lemma~\ref{la:refinement lifted to sigma},
		$\lambda^{\Obs}(q, \rho) = \lambda^{\Spec}(s,\rho) \neq \lambda^{\Spec}(t,\rho) = \lambda^{\Obs}(r, \rho)$.
		Thus $\rho \vdash r \apart q$.
	\end{enumerate} 
	Since all conditions of Corollary~\ref{Cy:k-A-complete} hold, we conclude that $T$ is $k$-$A$-complete.
\end{proofappendix}

\begin{proofappendix}{k-complete}
	Let $\M$ be a Mealy machine with at most $m$ states that passes suite $T$.
	In order to prove that $T$ is $m$-complete, it suffices to prove that $\M \approx \Spec$.
	Let $f : \Obs \to \Spec$ and $g : \Obs \to \M$.  Define relation $R \subseteq Q^{\Spec} \times Q^{\M}$ by
	\begin{eqnarray*}
		(s, q) \in R & \Leftrightarrow & \exists t \in B \cup  F^{<k} : f(t) = s \wedge  g(t) = q.
	\end{eqnarray*}
	We claim that $R$ is a bisimulation between $\Spec$ and $\M$.
	\begin{enumerate}
		\item 
		Since $f$ is a functional simulation from $\Obs$ to $\Spec$, $f(q^{\Obs}_0) = q^{\Spec}_0$, and
		since $g$ is a functional simulation from $\Obs$ to $\M$, $g(q^{\Obs}_0) = q^{\M}_0$.
		Using $q^{\Obs}_0 \in B$, this implies $(q^{\Spec}_0, q^{\M}_0) \in R$.
		\item 
		Suppose $(s,q) \in R$ and $i \in I$. 
		Since $(s, q) \in R$, there exists a $t \in B \cup F^{<k}$ such that $f(t) = s$ and $g(t) = q$.
		W.l.o.g.\  we select $t$ from the lowest possible stratum, that is, if $t \in F^i$ then there is no $\overline{t} \in B \cup F^{<i}$ with $f(\overline{t}) = s$ and $g(\overline{t}) = q$.
		Since $B$ and $F^{<k}$ are complete, $\delta^{\Obs}(t,i) \downarrow$.
		Since $f$ and $g$ are functional simulations, also $\delta^{\Spec}(s,i) \downarrow$ and $\delta^{\M}(q,i) \downarrow$.
		Let $s' = \delta^{\Spec}(s,i)$, $q' = \delta^{\M}(q,i)$ and $t'= \delta^{\Obs}(t,i)$.
		Since $f$ and $g$ are functional simulations,
		$\lambda^{\Obs}(t,i) = \lambda^{\Spec}(s,i)$ and $\lambda^{\Obs}(t,i) = \lambda^{\M}(q,i)$.
		This implies $\lambda^{\Spec}(s,i) = \lambda^{\M}(q,i)$, as required.
		Since $f$ and $g$ are functional simulations, $f(t') = s'$ and $g(t') = q'$.
		In order to prove $(s', q') \in R$, we consider two cases:
		\begin{enumerate}
			\item 
			$t' \in B \cup F^{<k}$.  In this case, since $f(t') = s'$ and $g(t') = q'$, $(s', q') \in R$ follows from the definition of $R$.
			\item 
			$t'\in F^k$.
			Let $W$ be the subset of states in $F^{<k}$ that occur in the access path of $t'$.
			We claim that $g$ is injective on $B \cup W$.
			Then, since $B \cup W$ has $n+k=m$ states, there must be a state $t'' \in B \cup W$ with $g(t'') = g(t')$.
			States $t''$ and $t'$ have the same candidate set (otherwise they would be apart by condition (\ref{cotransitivity requirement IM}) or by the fact that $t'$ has been identified, which would contradict $g(t'') = g(t')$). Then by Lemma~\ref{candidate set}, $f(t'') = f(t')$.
			This in turn implies that $(s', q') \in R$, which completes the proof that $R$ is bisimulation.
			
			Thus it remains to prove our claim that $g$ is injective on $B \cup W$:
			\begin{enumerate}
				\item 
				By Lemma~\ref{f restricted to B injective}, $g$ is injective on $B$.
				\item
				Let $u \in W$ and let $r \in B$.
				We claim $g(u) \neq g(r)$.	
				We consider two cases:
				\begin{itemize}
					\item 
					$C(u) = \{ r \}$.   Then by Lemma~\ref{candidate set}, $f(u) = f(r)$.
					But then $g(u) \neq g(r)$, because otherwise $t$ would not be in the lowest possible stratum.
					\item 
					$C(u) \neq \{ r\}$.  Then $u \apart r$ and therefore,
					by Lemma~\ref{la: apartness refinement}, $g(u) \neq g(r)$.
				\end{itemize}
				\item
				Let $u, u' \in W$ with $u \neq u'$.
				We claim $g(u) \neq g(u')$.
				By condition (\ref{cotransitivity requirement IM}), $C(u) = C(u')$ or $u\apart u'$.
				\begin{itemize}
					\item 
					If $C(u) = C(u')$, Lemma~\ref{candidate set} gives that $f(u) = f(u')$. 
					But then $g(u) \neq g(u')$, because otherwise $t$ would not be in the lowest possible stratum.
					\item 
					If $u \apart u'$ then $g(u) \neq g(u')$ by Lemma~\ref{la: apartness refinement}. 
				\end{itemize}
			\end{enumerate}
		\end{enumerate}
	\end{enumerate}
	The proposition now follows by application of Lemma~\ref{la:bisimulation}.
\end{proofappendix}

\end{document}